\documentclass{vldb}
\usepackage{mdframed}
\usepackage{mathtools}
\usepackage{amssymb}
\usepackage{graphicx}
\usepackage{framed}
\usepackage{tikz}
\usepackage[edges]{forest}
\usepackage[labelfont=bf]{caption}
\usepackage{subcaption}
\usepackage{float}
\usepackage{enumitem}
\usepackage{colortbl}
\usepackage{listings}
\usepackage{array}
\usepackage{xfrac}
\usepackage{multirow}
\usepackage[algo2e,ruled,vlined,linesnumbered]{algorithm2e}
\usepackage{pgfplots}
\usepackage{pgfplotstable}
\usetikzlibrary{pgfplots.groupplots}
\usetikzlibrary{shapes.geometric}
\usetikzlibrary{positioning}
\usetikzlibrary{backgrounds}
\usetikzlibrary{patterns}
\usepackage[hidelinks]{hyperref}
\hypersetup{colorlinks=false,}
\usepackage{pifont}
\usepackage{appendix}
\usepackage[normalem]{ulem}
\usepackage{siunitx}
\usepackage{balance}

\usepackage{times}
\usepackage[scaled=0.86]{helvet}

\pgfplotsset{compat=1.14}

\definecolor{colgreen}{RGB}{ 17,128,127}

\newcommand\Tstrut{\rule{0pt}{2.6ex}}       

\newtheorem{theorem}{Theorem}
\newdef{definition}{Definition}
\newdef{example}{Example}
\newdef{property}{Property}
\newtheorem{lemma}[theorem]{Lemma}

\definecolor{codegray}{rgb}{0.5,0.5,0.5}
\lstdefinestyle{pseudocode}{
  language=c,
  showspaces=false,
  basicstyle={\ttfamily \scriptsize},
  numberstyle=\tiny\color{codegray},
  keywordstyle=\color{magenta},
  numbers=left,
  numbersep=5pt,
  mathescape=true,
  numberblanklines=false,
  morekeywords={uint8_t,uint16_t,int16_t,in},
}
\lstdefinestyle{pcode}{
  style=pseudocode,
  basicstyle={\ttfamily \normalsize},
}
\lstdefinestyle{pcodesmall}{
  style=pseudocode,
  basicstyle={\ttfamily \small},
}
\newcommand{\pcode}[2][]{\textnormal{\lstinline[style=pcode,#1]|#2|}}
\newcommand{\pcodes}[2][]{\textnormal{\lstinline[style=pseudocode,#1]|#2|}}

\lstdefinestyle{SQLf}{
  language=SQL,
  showspaces=false,
  basicstyle={\ttfamily \footnotesize},
  morekeywords={REFERENCES},
  deletekeywords={YEAR},
  escapeinside={<@}{@>},
  mathescape=true,
}
\lstdefinestyle{xquery}{
  language=python,
  showspaces=false,
  basicstyle={\ttfamily \footnotesize},
  numberstyle=\tiny\color{codegray},
  keywordstyle=\color{magenta},
  morekeywords={for,where,and},
  numbersep=5pt,
  numberblanklines=false,
  mathescape=true,
  escapeinside={(@}{@)}
}

\forestset{
  doctree/.style={
    for tree={
      black,
      align=center,
      font={\ttfamily \small},
      transform shape,
      inner ysep=2pt,
      l=0mm,
      l sep-=2mm,
    },
  },
  bomtree/.style={
    doctree,
    for tree={
      s sep-=1mm,
      l sep+=0.5mm,
    },
    for tree={scale=0.8},
  },
}

\newcommand{\pstr}[1]{\textcolor{blue}{#1}}

\newcommand{\ptstr}[1]{\texttt{\textcolor{blue}{#1}}}
\newcommand{\vtstr}[1]{\texttt{\textcolor{red}{\underline{#1}}}}
\newcommand{\ptbstr}[1]{\textbf{\texttt{\textcolor{blue}{#1}}}}
\newcommand{\vtbstr}[1]{\textbf{\texttt{\textcolor{red}{\underline{#1}}}}}

\forestset{
  casindex/.style={
    for tree={
      font={\ttfamily \small},
      align=center,
      parent anchor=south,
      child anchor=north,
      s sep-=2mm,
    },
  },
  declare toks={elo}{
    font={\ttfamily \small}
  },
  dot/.style={
    tikz+={
      \fill[#1] (.child anchor) circle[radius=2.5pt];
    }
  },
  etype/.style n args={3}{
    edge path={
      \noexpand \path[\forestoption{edge},#3]
        (!u.parent anchor) -- (.child anchor)
        \forestoption{edge label};
      \noexpand \path[#1] (.child anchor) circle[radius=2pt];
      \noexpand \path[#2] (!u.parent anchor) circle[radius=2pt];
    },
  },
  ppnode/.style={
    etype={fill=black}{draw=none}{black},
  },
  pvnode/.style={
    etype={fill=black}{draw=none}{black},
  },
  vpnode/.style={
    etype={fill=black}{draw=none}{black},
  },
  vvnode/.style={
    etype={fill=black}{draw=none}{black},
  },
  lpnode/.style={
    etype={fill=black}{draw=none}{black},
  },
  lvnode/.style={
    etype={fill=black}{draw=none}{black},
  },
  bbnode/.style={
    etype={fill=black}{draw=none}{black},
  },
  lbnode/.style={
    etype={fill=black}{draw=none}{black},
  },
  elabel/.style n args=2{
    edge label/.expanded={
      node[midway,anchor=#1,\forestoption{elo}]{\strut\unexpanded{#2}}
    }
  },
}

\forestset{
  searchindex/.style={
    casindex,
    for tree={
      parent anchor=north,
      l=0mm,
      s sep+=5pt,
      bbnode,
    },
  },
}

\pgfplotsset{
  rcas/.style={
    color=colxb,
    mark=o,
  },
  casxml/.style={
    color=colxc,
    mark=triangle*,
  },
  caspv/.style={
    color=colx6,
    mark=diamond,
  },
  caspvs/.style={
    color=colx6,
    mark=diamond*,
  },
  casvp/.style={
    color=colxa,
    mark=square,
  },
  casvps/.style={
    color=colxa,
    mark=square*,
  },
  casbw/.style={
    color=colx1,
    mark=o,
  },
  caslw/.style={
    color=colxd,
    mark=triangle,
  },
  caszo/.style={
    color=orange,
    mark=triangle,
  },
  casseq/.style={
    color=colx4,
    mark=x,
  },
  lineModel/.style={
    color=colx9,
    mark=o,
  },
  experiment/.style={
    height=45mm,
    width=80mm,
    ticklabel style={font=\small},
    xlabel/.append style={font=\small},
    ylabel/.append style={font=\small},
    legend columns=4,
    legend style={
      at={(0.0,1.03)},
      anchor=south west,
      draw=none,
      cells={anchor=west},
    },
  },
  selectivity x axis/.style={
    xticklabel style={
      /pgf/number format/fixed,
      /pgf/number format/precision=2
    },
  },
}

\pgfplotsset{
  idy/.style={
    red,
    fill=white,
    postaction={
      pattern=crosshatch dots,
      pattern color=red,
    },
  },
  ipv/.style={
    blue,
    fill=white,
    postaction={
      pattern=north east lines,
      pattern color=blue,
    },
  },
  ivp/.style={
    blue,
    fill=white,
    postaction={
      pattern=north west lines,
      pattern color=blue,
    },
  },
  ione/.style={
    brown!60!black,
    fill=white,
    postaction={
      pattern=horizontal lines,
      pattern color=brown!60!black,
    },
  },
  itwo/.style={
    brown!60!black,
    fill=white,
    postaction={
      pattern=crosshatch,
      pattern color=brown!60!black,
    },
  },
}

\pgfplotsset{
  barRCAS/.style={
    colxb,
    fill=white,
    mark=none,
    postaction={
      pattern=crosshatch dots,
      pattern color=colxb,
    },
  },
  barZO/.style={
    colxc,
    fill=white,
    mark=none,
    postaction={
      pattern=north east lines,
      pattern color=colxc,
    },
  },
  barLW/.style={
    colxd,
    fill=white,
    mark=none,
    postaction={
      pattern=north west lines,
      pattern color=colxd,
    },
  },
  barPV/.style={
    colx6,
    fill=white,
    mark=none,
    postaction={
      pattern=horizontal lines,
      pattern color=colx6,
    },
  },
  barVP/.style={
    colxa,
    fill=white,
    mark=none,
    postaction={
      pattern=crosshatch,
      pattern color=colxa,
    },
  },
  barXML/.style={
    colxc,
    fill=white,
    mark=none,
    postaction={
      pattern=dots,
      pattern color=colxc,
    },
  },
  barP/.style={
    colxh!40!black,
    fill=white,
    mark=none,
    postaction={
      pattern=north east lines,
      pattern color=colxh!60!black,
    },
  },
  barV/.style={
    red,
    fill=white,
    mark=none,
    postaction={
      pattern=crosshatch,
      pattern color=red,
    },
  },
  barL/.style={
    colxf!40!black,
    fill=white,
    mark=none,
    postaction={
      pattern=crosshatch dots,
      pattern color=colxf!60!black,
    },
  },
  barModel/.style={
    colx9,
    fill=white,
    mark=none,
    postaction={
      pattern=north east lines,
      pattern color=colx9,
    },
  },
  barStructure/.style={
    colxg!30!black,
    fill=colxg,
    mark=none,
  },
}

\pgfplotsset {
  col2/.style={
    height=33mm,
    width=50mm,
    ticklabel style={font=\scriptsize},
    xlabel near ticks,
    ylabel near ticks,
    ylabel/.append style={align=center},
    ylabel/.append style={font=\scriptsize},
    xlabel/.append style={font=\scriptsize},
  },
  table/readtable/.style={
    x=value_selectivity,
    y expr=\thisrow{runtime_mus}/1000,
    col sep=comma,
  },
  table/tabstyle/.style={
    x=value_selectivity,
    col sep=comma,
  },
}

\definecolor{colx1}{RGB}{252, 26,135}  
\definecolor{colx2}{RGB}{ 11, 36,251}  
\definecolor{colx3}{RGB}{ 17,128,127}  
\definecolor{colx4}{RGB}{183,108, 50}  
\definecolor{colx5}{RGB}{253,128, 35}  
\definecolor{colx6}{RGB}{127, 15,126}  
\definecolor{colx7}{RGB}{200,  0, 20}  
\definecolor{colx8}{RGB}{191,191,191}  
\definecolor{colx9}{RGB}{100,  0,  4}  
\definecolor{colxa}{RGB}{  0,139, 60}  
\definecolor{colxb}{RGB}{  0, 98,139}  
\definecolor{colxc}{RGB}{242,112, 36}  
\definecolor{colxd}{RGB}{222, 23, 23}  
\definecolor{colxe}{RGB}{211,211,211}  
\definecolor{colxf}{RGB}{236,217,198}  
\definecolor{colxg}{RGB}{ 70,192,111}  
\definecolor{colxh}{RGB}{  0,185,242}  

\makeatletter
\renewcommand{\tableofcontents}{\scriptsize\@starttoc{toc}}
\makeatother

\pgfplotsset{compat=1.14,
  /pgfplots/ybar legend/.style={
    /pgfplots/legend image code/.code={%
      \draw[##1,/tikz/.cd,yshift=-0.25em]
      (0cm,0cm) rectangle (7pt,1.1em);
    },
  },
}

\newcommand{\rev}[1]{#1}

\begin{document}

\title{Dynamic Interleaving of Content and Structure for\\Robust
Indexing of  Semi-Structured Hierarchical Data (Extended Version)}
\vldbTitle{Dynamic Interleaving of Content and Structure for Robust
Indexing of  Semi-Structured Hierarchical Data}

\vldbAuthors{Kevin Wellenzohn, Michael H.~B{\"o}hlen, Sven Helmer}
\vldbDOI{https://doi.org/10.14778/3401960.3401963}
\vldbNumber{10}
\vldbVolume{13}
\vldbYear{2020}

\numberofauthors{3}
\author{
\alignauthor
Kevin Wellenzohn \\
       \affaddr{University of Zurich}\\
       \email{\href{mailto:wellenzohn@ifi.uzh.ch}{wellenzohn@ifi.uzh.ch}}
\and
\alignauthor
Michael H. B{\"o}hlen \\
       \affaddr{University of Zurich}\\
       \email{\href{mailto:boehlen@ifi.uzh.ch}{boehlen@ifi.uzh.ch}}
\and
\alignauthor
Sven Helmer\\
       \affaddr{University of Zurich}\\
       \email{\href{mailto:helmer@ifi.uzh.ch}{helmer@ifi.uzh.ch}}
}

\additionalauthors{}
\pagenumbering{arabic}

\maketitle

\begin{abstract}
  We propose a robust index for semi-structured hierarchical data that
  supports content-and-structure (CAS) queries specified by path and
  value predicates. At the heart of our approach is a novel dynamic
  interleaving scheme that merges the path and value dimensions of
  composite keys in a balanced way.  We store these keys in our
  trie-based Robust Content-And-Structure index, which efficiently
  supports a wide range of CAS queries, including queries with
  wildcards and descendant axes. Additionally, we show important
  properties of our scheme, such as robustness against varying
  selectivities, and demonstrate improvements of up to two orders of
  magnitude over existing approaches in our experimental evaluation.
\end{abstract}


\section{Introduction}
\label{sec:intro}

A lot of the data in business and engineering applications is
semi-structured and inherently hierarchical.  Typical examples are
bills of materials (BOMs) \cite{RB15}, enterprise asset hierarchies
\cite{JF13}, and enterprise resource planning applications
\cite{JF15}.  A common type of queries on such data are
content-and-structure (CAS) queries \cite{CM15}, containing a
\emph{value predicate} on the \emph{content} of some attribute and a
\emph{path predicate} on the location of this attribute in the
\emph{hierarchical structure}.

As real-world BOMs grow to tens of millions of nodes \cite{JF13}, we
need dedicated CAS access methods to support the efficient processing
of CAS queries.  Existing CAS indexes often lead to large intermediate
results, since they either build separate indexes for, respectively,
content and structure \cite{CM15} or prioritize one dimension over the
other (i.e., content over structure or vice versa)
\cite{oak19,BC01,JL15}.  We propose a \emph{well-balanced integration}
of paths and values in a single index that provides \emph{robust}
performance for CAS queries, meaning that the index prioritizes
neither paths nor values.

We achieve the balanced integration of the path and value dimension
with composite keys that \emph{interleave} the bytes of a path and a
value. Interleaving is a well-known technique applied to
multidimensional keys, for instance Nishimura et al.~look at a family
of bit-merging functions \cite{SN17} that include the $c$-order
\cite{SN17} and the $z$-order \cite{GM66,JO84} space-filling curves.
Applying space-filling curves on paths and values is subtle, though,
and can result in poor query performance because of varying key
length, different domain sizes, and the skew of the data.  The
$z$-order curve, for example, produces a poorly balanced partitioning
of the data if the data contains long common prefixes \cite{VM99}.  The
paths in a hierarchical structure exhibit this property: they have, by
their very nature, long common prefixes. The issue with common
prefixes is that they do not help to partition the data, since they
are the same for all data items. However, the first byte following a
longest common prefix does exactly this: it distinguishes different
data items. We call such a byte a \emph{discriminative byte}.  The
distribution of discriminative path and value bytes in an interleaved
key determines the order in which an index partitions the data and,
consequently, how efficiently queries can be evaluated.  The $z$-order
of a composite key often clusters the discriminative path and value
bytes, instead of interleaving them. This leads to one dimension---the
one whose discriminative bytes appear first---to be prioritized over
the other, precluding robust query performance.

We develop a \emph{dynamic interleaving} scheme that interleaves the
discriminative bytes of paths and values in an alternating way.  This
leads to a well-balanced partitioning of the data with a robust query
performance.  Our dynamic interleaving is \emph{data-driven} since the
positions of the discriminative bytes depend on the distribution of
the data.  We use the dynamic interleaving to define the \emph{Robust
  Content-and-Structure (RCAS)} index for semi-structured hierarchical
data.  \rev{We build our RCAS index as an \rev{in-memory} trie
  data-structure \cite{VL13} to efficiently support the basic search
  methods for CAS queries: \emph{range searches} and \emph{prefix
    searches}.  Range searches enable value predicates that are
  expressed as a value range and prefix searches allow for path
  predicates that contain wildcards and descendant axes.} Crucially,
tries in combination with dynamically interleaved keys allow us to
efficiently evaluate path and value predicates simultaneously.
\rev{We provide an efficient bulk-loading algorithm for RCAS that
  scales linearly with the size of the dataset.  Incremental
  insertions and deletions are not supported.}

Our main contributions can be summarized as follows:
\begin{itemize}

\item We develop a \emph{dynamic interleaving} scheme to interleave
  paths and values in an alternating way using the concept of
    \emph{discriminative bytes}.  We show how to compute this
    interleaving by partitioning the data. We prove that our dynamic
    interleaving is robust against varying selectivities
    (Section~\ref{sec:dynint}).

\item We propose the \rev{in-memory}, trie-based \emph{Robust
  Content-and-Structure (RCAS) index} for semi-structured hierarchical
    data. The RCAS achieves its robust query performance by a
    well-balanced integration of paths and values via our dynamic
    interleaving scheme (Section \ref{sec:rcas}).

\item Our RCAS index supports a broad spectrum of CAS queries that
  include wildcards and the descendant axis.  We show how to evaluate
    CAS queries through a combination of range and prefix searches on
    the trie-based structure of the RCAS index (Section
    \ref{sec:querying}).

\item An exhaustive experimental evaluation with real-world and
  synthetic datasets shows that RCAS delivers robust query
    performance.  We get improvements of up to two orders of magnitude
    over existing approaches (Section \ref{sec:experiments}).
\end{itemize}


\section{Running Example}
\label{sec:rexample}


\rev{We consider a company that stores the bills of materials (BOMs)
  of its products. BOMs represent the hierarchical assembly of
  components to final products.  Each BOM node is stored as a tuple in
  a relational table, which is common for hierarchies, see, e.g.,
  SAP's storage of BOMs \cite{RB15,JF13,JF15} and the Software
  Heritage Archive \cite{RC17,AP19}.  A CAS index is used to
  efficiently answer queries on the structure (location of a node in
  the hierarchy) and the content of an attribute (e.g., the weight or
  capacity).  The paths of all nodes in the BOM that have a value for
  the indexed attribute as well as the value itself are indexed in the
  CAS index.  The index is read-only, updated offline, and kept in
  main memory.}

Figure \ref{fig:bom} shows the hierarchical representation of a BOM
for three products.  The components of each product are organized
under an \texttt{item} node.  Components can have attributes to record
additional information, e.g., the weight of a battery.  Attributes are
represented by special nodes that are prefixed with an \texttt{@} and
that have an additional value.  For example, the weight of the
rightmost battery is
250'714 grams and its capacity is 80000 Wh.

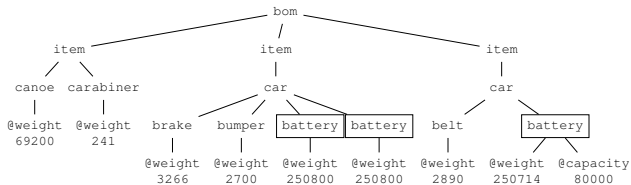
\begin{figure}[htb]
  \centering
  \scalebox{.9}{
  {\begin{forest}
    emptypath/.style={
      edge=dotted,
    },
    bomtree,
    for tree={
      s sep-=1mm,
      font={\ttfamily \scriptsize},
    },
    [bom
      [item
        [canoe
          [@weight\\69200
          ]
        ]
        [carabiner
          [@weight\\241
          ]
        ]
      ]
      [item
        [car
          [brake
            [@weight\\3266
            ]
          ]
          [bumper
            [@weight\\2700
            ]
          ]
          [battery,draw
            [@weight\\250800
            ]
          ]
          [battery,draw
            [@weight\\250800
            ]
          ]
        ]
      ]
      [item
        [car
          [belt
            [@weight\\2890
            ]
          ]
          [battery,draw
            [@weight\\250714
            ]
            [@capacity\\80000
            ]
          ]
        ]
      ]
    ]
  \end{forest}}}
  \caption{Example of a bill of materials (BOM).}
  \label{fig:bom}
\end{figure}

Next, we look at an example CAS query. We roughly follow the syntax of
query languages for semistructured data, such as XQuery~\cite{Katz03}
or JSONiq~\cite{DF13}, utilizing simple FLWOR expressions.

\begin{example}
  To reduce the weight of cars we look for all heavy car parts, i.e.,
  parts weighing at least 50 kilograms (``\texttt{//}'' matches a node
  and all its descendants in a hierarchical structure):

  \medskip \hspace{-15pt} {\footnotesize
    \begin{tabular}{ll}
      $Q$:
      & \texttt{\textbf{for} \$c \textbf{in}
        {\color{blue}/bom/item/car//}} \\
      & \texttt{\textbf{where} \$c/{\color{red}@weight >= 50000}} \\
      & \texttt{\textbf{return} \$c} \\
    \end{tabular}}
\end{example}

The answer to query $Q$ are the three framed nodes
in Figure~\ref{fig:bom}.  Our goal is an index
that guides us as quickly as possible to these nodes.  Indexes on
either paths or values do not perform well.  An index built for only
the values of weight nodes also accesses the node for the canoe.  A
purely structural index for the paths additionally has to look at
the weight of other car parts.  Our RCAS index considers values and
paths together to get a good query performance.


\section{Related Work}
\label{sec:rw}

We begin with a review of existing CAS indexes
\cite{BC01,RK04,HL06,CM15,JL15}.  IndexFabric \cite{BC01} prioritizes
the structure of the data over its values.  It stores the concatenated
path and value of a key in a disk-optimized PATRICIA trie \cite{DM68}
\rev{that supports incremental updates (i.e., inserts and deletes)}.
IndexFabric does not offer robust CAS query performance since a CAS
query must fully evaluate a query's path predicate before it can
evaluate its value predicate.  This leads to large intermediate
results if the path predicate is not selective.

The hierarchical database system Apache Jackrabbit Oak \cite{oak19}
implements a CAS index that prioritizes values over paths. Oak indexes
(value $v$, path $p$)-pairs in a DataGuide-like index \rev{that
  supports updates}.  For each value $v$, Oak stores a DataGuide
\cite{RG97} of all paths $p$ that have this particular value $v$.
Query performance is poor if the value predicate is not selective
because the system must search many DataGuides.

The CAS index of Microsoft Azure's DocumentDB (now Cosmos DB)
concatenates paths and values \cite{JL15} and stores the result in a
Bw-tree \cite{JL13} \rev{that supports updates}.  Depending on the
expected query type(s) (point or range queries), the system either
stores forward keys (e.g., \texttt{/a/b/c}) or reverse keys (e.g.,
\texttt{c/b/a}). To reduce the space requirements, forward and
reverse keys are split into trigrams (three consecutive node labels).
During the evaluation of a CAS query these trigrams must be joined and
matched against the query, which is slow.  Moreover, choosing forward
or reverse keys prioritizes structure over values or
vice-versa.

Mathis et al.~\cite{CM15} propose a CAS index that consists of two
index structures: a B-tree to index the values and a structural
summary (e.g., a DataGuide \cite{RG97}) to index the structure of the
data.  The DataGuide assigns an identifier (termed PCR) to each
distinct path in the documents. The B-tree stores the values along
with the PCRs.  The path and value predicates of a CAS query are
independently evaluated on the DataGuide and the B-tree, and the
intermediate results are joined on the PCR.  This is expensive if the
intermediate results are large (i.e., at least one predicate is not
selective) but the final result is small.  \rev{Updates are supported
  and are executed on the B-tree as well as the DataGuide.}

Kaushik et al.~\cite{RK04} present an approach that joins inverted
lists for answering CAS queries.  They use a 1-index \cite{TM99} to
evaluate path predicates and B-trees to evaluate value predicates.
This approach evaluates both predicates independently and exhibits the
same problems as \cite{CM15}. \rev{Updates are not
discussed.}

FLUX \cite{HL06} computes a Bloom filter for each path into which its
labels are hashed and stores these Bloom filters along with the values
in a B-tree. Query evaluation proceeds as follows.  The value
predicate is matched on the B-tree and for each matched value the
corresponding Bloom filter $C$ is compared to a Bloom filter $Q$
built for the query path. If each bit that is set in $Q$ is also set
in $C$, the path is a possible match that needs to be double-checked
through database accesses.  Value predicates that are not selective
produce large intermediate results. \rev{Updates are not discussed.}

Some document databases for semi-structured, hierarchical data
(MongoDB \cite{mongo}, CouchDB \cite{couch}, and AsterixDB
\cite{SA14}) use pure value indexes (e.g., standard B-trees) that
index the content of documents but not their structure. They create an
index on a predefined path (e.g., \texttt{/person/name}) and only
index the corresponding values. They cannot answer CAS queries with
arbitrary path predicates.

Besides pure value indexes there are also pure structure indexes
that focus on complex twig queries with different axes
(ancestor, descendant, sibling, etc.).  DeltaNI \cite{JF13} and Order
Indexes \cite{JF17} are recent proposals in this area.  Pure structure
indexes cannot efficiently answer CAS queries that also include value
predicates.

Our RCAS index integrates paths and values by interleaving them. This
is similar to the bit-merging family of space-filling curves that
combine the binary representation of a key's dimensions.  We compare
our approach to two representatives: the $c$-order curve \cite{SN17}
and the $z$-order curve \cite{GM66,JO84}.  The $c$-order curve is
obtained by concatenating dimensions, which prioritizes one of the
dimensions.  The selectivity of the predicate on the prioritized
dimension determines the query performance. If it is high and the
other selectivity is low, the $c$-order curve performs badly.  The
$z$-order curve is a space-filling curve that is used, among others,
by UB-trees \cite{FR00} and k-d tries \cite{BN08,JO84,HS06}.  It is
obtained by the bit-wise interleaving of dimensions.  The $z$-order
curve produces an unbalanced partitioning of the data with poor query
performance if the data contains long common prefixes.  Markl calls
this the ``puff-pastry effect'' \cite{VM99} because the query
performance deteriorates to that of a $c$-order curve that fully
orders one dimension after another.  The Variable UB-tree \cite{VM99}
uses a pre-processing step to encode the data in such a way that the
puff-pastry effect is avoided.  The encoding is not prefix-preserving
and cannot be used in our CAS index. We need prefix searches to
evaluate path predicates.  \rev{The $c$-order and $z$-order curves are
  \emph{static} interleaving schemes that do not take the data
  distribution into account.  Indexes based on static schemes can be
  updated efficiently since insertions and deletions do not affect
  existing interleavings.  The Variable UB-tree does not support
  incremental updates \cite{VM99} since its encoding function adapts
  to the data distribution and must be recomputed whenever the data
  changes.  Similarly, our data-driven dynamic interleaving does not
  support incremental updates since the position of the discriminative
  bytes may change when keys are inserted or deleted.}

QUILTS \cite{SN17} devises a \rev{static} interleaving scheme for a
specific query workload. \rev{Index updates, although not discussed,
  would work as for other static schemes (e.g., $c$- and $z$-order).}
Our dynamic interleaving adapts to the data distribution rather than a
specific query workload.  We do not optimize a specific workload but
want to support a wide range of queries in a robust way, including
ad-hoc queries.


\section{Background}
\label{sec:background}

\textbf{Composite Keys.} We use composite keys that consist of a path
dimension $P$ and value dimension $V$ to index attributes in
hierarchical data.  Neither paths nor values nor the combination
of paths and values need to be unique in a database. Composite keys
can be extracted from popular semi-structured hierarchical data
formats, such as JSON and XML.

\begin{definition}(\emph{Composite Key}) %
  A composite key $k$ states that a node with path $k.P$ in a database
  has value $k.V$.
\end{definition}

Let $D \in \{P, V\}$ be the path or value dimension. We write $k.D$ to
access $k$'s path (if $D = P$) or value (if $D = V$).
The value dimension can be of any
primitive data type.  In the remainder of this paper we use one byte
ASCII characters for the path dimension and hexadecimal numbers for
the value dimension.

\begin{example}
  In our running example we index attribute \texttt{@weight}.
  Table~\ref{tab:ex2} shows the composite keys for the
  \texttt{@weight} attributes from the BOM in Figure \ref{fig:bom}.
  Since only the \texttt{@weight} attribute is indexed, we omit the
  label \texttt{@weight} in the paths in Table~\ref{tab:ex2}. The
  values of the \texttt{@weight} attribute are stored as 32 bit
  unsigned integers.
\end{example}

The set of composite keys in our running example is denoted by
$\mathsf{K}^{1..7} = \{\mathsf{k}_1, \mathsf{k}_2, \ldots,
\mathsf{k}_7 \}$, see Table \ref{tab:ex2}.  We use a
\textsf{sans-serif} font to refer to concrete values. Further, we use
notation $\mathsf{K}^{2,5,6,7}$ to refer to $\{\mathsf{k}_2,
\mathsf{k}_5, \mathsf{k}_6, \mathsf{k}_7\}$.

\begin{table}[htb] \centering
\caption{A set $\mathsf{K}^{1..7} = \{ \mathsf{k}_1, \ldots,
\mathsf{k}_7 \}$ of composite
keys. The values are stored as 32 bit unsigned integers.}
\label{tab:ex2}
\begin{tikzpicture}[
    tick/.style={
      font={\ttfamily \scriptsize},
      anchor=north,
    },
  ]
  \node[inner sep=0, outer sep=0,anchor=south west] (table) at (0,0) {
  \begin{tabular}{l|lc|}
    \cline{2-3}
    \multicolumn{1}{c|}{}
      & \multicolumn{1}{c}{Path Dimension $P$}
      & \multicolumn{1}{c|}{Value Dimension $V$} \\
    \cline{2-3}
    $\mathsf{k}_1$
      & \texttt{/bom/item/canoe\$}
      & \texttt{00\,01\,0E\,50} \\
    $\mathsf{k}_2$
      & \texttt{/bom/item/carabiner\$}
      & \texttt{00\,00\,00\,F1} \\
    $\mathsf{k}_3$
      & \texttt{/bom/item/car/battery\$}
      & \texttt{00\,03\,D3\,5A} \\
    $\mathsf{k}_4$
      & \texttt{/bom/item/car/battery\$}
      & \texttt{00\,03\,D3\,B0} \\
    $\mathsf{k}_5$
      & \texttt{/bom/item/car/belt\$}
      & \texttt{00\,00\,0B\,4A} \\
    $\mathsf{k}_6$
      & \texttt{/bom/item/car/brake\$}
      & \texttt{00\,00\,0C\,C2} \\
    $\mathsf{k}_7$
      & \texttt{/bom/item/car/bumper\$}
      & \texttt{00\,00\,0A\,8C} \\
    \cline{2-3}
  \end{tabular}
  };
  \foreach \x in {0,2,...,22} {
    \pgfmathsetmacro{\xc}{\x * 0.190 + 1.0}
    \pgfmathtruncatemacro{\xr}{\x + 1}
    \draw (\xc,0) -- (\xc,-0.1);
    \node[tick] at (\xc,-0.05) {\rev{\xr}};
  }
  \foreach \x in {0,...,3} {
    \pgfmathsetmacro{\xc}{\x * 0.43 + 6.05}
    \pgfmathtruncatemacro{\xr}{\x + 1}
    \draw (\xc,0) -- (\xc,-0.1);
    \node[tick] at (\xc,-0.05) {\rev{\xr}};
  }
\end{tikzpicture}
\end{table}

\smallskip \textbf{Querying.} Content-and-structure (CAS) queries
contain a path predicate and value predicate \cite{CM15}. \rev{The
path predicate is expressed as a query path $q$ that may include
\texttt{//} to match a node itself and all its descendants, and the
wildcard \texttt{*} to match all of a node's children. The latter is
useful for data integrated from sources using different terminology
(e.g., \texttt{product} instead of \texttt{item} in Fig.~\ref{fig:bom}).}

\begin{definition}(\emph{Query Path}) \label{def:querypath}
  A query path $q$ is denoted by
  $q = e_1\,\lambda_1\,e_2\,\lambda_2 \,\ldots\,\lambda_{m-1}\,e_m$.
  Each $e_i$, $i \leq m$, is either the path separator \texttt{/} or
  the descendant-or-self axis \texttt{//} that matches zero to any
  number of descendants.  The final path separator $e_m$ is optional.
  $\lambda_i$, $i < m$, is either a label or the wildcard
  \texttt{*} that matches any label.
\end{definition}

\begin{definition}(\emph{CAS Query})
  A CAS query $Q(q,\theta)$ consists of a query path $q$ and a value
  predicate $\theta$ on an attribute $A$, where  $\theta$ is a simple
  comparison $\theta = A \,\Theta\, v$ or a range comparison $\theta =
  v_l \,\Theta\, A \,\Theta'\, v_h$ where $\Theta, \Theta' \in
  \{=,<,>,\leq,\ge\}$.  Let $K$ be a set of composite keys.  CAS query
  $Q$ returns all composite keys $k \in K$ such that $k.P$ satisfies
  $q$ and $k.V$ satisfies $\theta$.
\end{definition}

\begin{example} \label{ex:casquery}
  The CAS query from Section \ref{sec:rexample} is expressed as
  $Q(\texttt{/bom/item/car//}, \texttt{@weight} \ge 50000)$ and
  returns all car parts weighing more than 50000 grams in Figure
  \ref{fig:bom}.  \rev{CAS query $Q(\texttt{/bom/*/car/battery},
  \texttt{@capacity} = 80000)$ looks for all car batteries
  that have a capacity of 80kWh. The wildcard $\texttt{*}$ matches any child
  of $\texttt{bom}$ (only \texttt{item} children exist in our example).}
\end{example}

\textbf{Representation of Keys.} Paths and values are prefix-free
byte strings as illustrated in Table~\ref{tab:ex2}.  To get
prefix-free byte strings we append the end-of-string character (ASCII
code \texttt{0x00}, here denoted by \texttt{\$}) to each path.  This
guarantees that no path is prefix of another path.  Fixed-length byte
strings (e.g., 32 bit numbers) are prefix-free because of the fixed
length.

Let $s$ be a byte-string, then $\textsf{len}(s)$ denotes the length of
$s$ and $s[i]$ denotes the $i$-th byte in $s$.  The left-most byte of
a byte-string is byte one.  $s[i] = \epsilon$ is the empty string if
$i > \textsf{len}(s)$. $s[i,j]$ denotes the substring of $s$ from
position $i$ to $j$ and $s[i,j] = \epsilon$ if $i > j$.

\medskip \textbf{Interleaving of Composite Keys.} We integrate path
$k.P$ and value $k.V$ of a key $k$ by interleaving them.
Figure~\ref{fig:interleaving} shows various interleavings of key
$\mathsf{k}_6$ from Table~\ref{tab:ex2}.  \rev{Value bytes are
underlined and shown in red, path bytes are shown in blue.} The first
two rows show the two $c$-order curves: path-value and value-path
concatenation ($I_{PV}$ and $I_{VP}$).  The byte-wise interleaving
$I_{BW}$ in the third row interleaves one value byte with one path
byte.  Note that none of these interleavings is well-balanced.  The
byte-wise interleaving is not well-balanced, since all value-bytes are
interleaved with parts of the common prefix of the paths
(\texttt{/bom/item/ca}).  In our experiments we use the
surrogate-based extension proposed by Markl~\cite{VM99} to more evenly
interleave dimensions of different lengths (see
Section~\ref{sec:experiments}).

\begin{figure}[htb]
\centering
\begin{tikzpicture}
\node[anchor=north west] (table) at (0,0) {
${\normalsize \begin{alignedat}{3}
  & \text{\small Approach} && && \hspace{1.7cm}\text{\small
  Interleaving of Key} \\
  \hline
  \rule{0pt}{2.6ex}
  & I_{PV}(\mathsf{k}_6) && = &&
    \ptstr{/bom/item/car/brake\$}\,
    \vtstr{00\,00\,0C\,C2}
    \\
  & I_{VP}(\mathsf{k}_6) && =\,\, &&
    \vtstr{00\,00\,0C\,C2}\,
    \ptstr{/bom/item/car/brake\$}
    \\
  & I_{BW}(\mathsf{k}_6) && = &&
    \vtstr{00}\,
    \ptstr{/}\,
    \vtstr{00}\,
    \ptstr{b}\,
    \vtstr{0C}\,
    \ptstr{o}\,
    \vtstr{C2}\,
    \ptstr{m}
    \ptstr{/item/car/brake\$}
    \\
  \hline
\end{alignedat}}$
};
\end{tikzpicture}
\caption{Key $\mathsf{k}_6$ is interleaved using
different approaches.}
\label{fig:interleaving}
\end{figure}


\section{Dynamic Interleaving}
\label{sec:dynint}

Our dynamic interleaving is a \emph{data-driven} approach to
interleave the paths and values of a set of composite keys $K$.  It
adapts to the specific characteristics of paths and values, such as
varying length, differing domain sizes, and the skew of the data.  To
this end, we consider the distribution of \emph{discriminative bytes}
in the indexed data.

\begin{definition}(\emph{Discriminative Byte}) \label{def:discbyte}
  The discriminative byte of a set of composite keys $K$ in dimension
  $D \in \{P,V\}$ is the position of the first byte in dimension $D$
  for which not all keys are equal:
  \begin{align*}
    \textsf{dsc}(K&,D) = m \text{ iff} \\
      & \exists k_i, k_j \in K, i \neq j
        ( k_i.D[m] \neq k_j.D[m]) \text{ and} \\
      & \forall k_i, k_j\in K, l<m
        (k_i.D[l] = k_j.D[l])
  \end{align*}
  If all
  values of dimension $D$ in $K$ are equal, the
  discriminative byte does not exist.  In this case we define
  $\textsf{dsc}(K, D) = \textsf{len}(k_i.D) + 1$ for some
  $k_i \in K$. $\hfill\Box$
\end{definition}

\begin{example}
  Table \ref{tab:discbyte} illustrates the position of the
  discriminative bytes for the path and value dimensions for various
  sets of composite keys $K$.
  \begin{table}[htbp] \centering
    \caption{Illustration of the discriminative bytes for
      $\mathsf{K}^{1..7}$ from Table \ref{tab:ex2} and various subsets
      of it.}
    \label{tab:discbyte}
    \begin{tabular}{l|cc}
      \multicolumn{1}{l|}{Composite Keys $K$}
        & \multicolumn{1}{c}{$\textsf{dsc}(K,P)$}
        & \multicolumn{1}{c}{$\textsf{dsc}(K,V)$}
        \\ \hline
      $\mathsf{K}^{1..7}$
        & 13
        & 2
        \Tstrut
        \\
      $\mathsf{K}^{2,5,6,7}$
        & 14
        & 3 \\
      $\mathsf{K}^{5,6,7}$
        & 16
        & 3 \\
      $\mathsf{K}^{6}$
        & 21
        & 5 \\
    \end{tabular}
  \end{table}
\end{example}

Discriminative bytes are crucial during query evaluation since at the
position of the discriminative bytes the search space can be narrowed
down.  We alternate in a round-robin fashion between discriminative
path and value bytes in our dynamic interleaving.  Note that in order
to determine the dynamic interleaving of a key $k$, which we denote by
$I_{\text{DY}}(k,K)$, we have to consider the set of keys $K$ to which
$k$ belongs and determine where the keys in $K$ differ from each other
(i.e., where their discriminative bytes are located).  Each
discriminative byte partitions $K$ into subsets, which we recursively
partition further.

\subsection{Partitioning by Discriminative Bytes}

The partitioning of a set of keys $K$ groups composite keys together
that have the same value for the discriminative byte in dimension $D$.
Thus, $K$ is split into at most $2^8$ non-empty partitions, one
partition for each value (\texttt{0x00} to \texttt{0xFF}) of the
discriminative byte of dimension $D$.

\begin{definition}(\emph{$\psi$-Partitioning})
  \label{def:partitioning}
  $\psi(K, D) = \{ K_1, \ldots, K_m\}$ is the
  $\psi$-partitioning of composite keys $K$ in dimension $D$
  iff all partitions are non-empty
    ($K_i \neq \emptyset \text{ for } 1 \leq i \leq m$), the
    number $m$ of partitions is minimal, and:
    \begin{enumerate} \setlength{\itemsep}{0pt plus 1pt}
    \item All keys in partition $K_i \in \psi(K,D)$ have the same
      value for the discriminative byte of $K$ in dimension $D$:
      \vspace{-3pt}
      \begin{itemize}[leftmargin=10pt]
      \item[--]
        $\forall k_u,k_v \in K_i \left( k_u.D[\textsf{dsc}(K,D)] =
          k_v.D[\textsf{dsc}(K,D)] \right) $
      \end{itemize}

    \item The partitions are disjoint:
    \vspace{-3pt}
      \begin{itemize}[leftmargin=10pt]
      \item[--]
        $\forall K_i,K_j \in \psi(K, D) (
        K_i \neq K_j \Rightarrow K_i \cap
          K_j = \emptyset)$
      \end{itemize}

    \item The partitioning is complete:
      \vspace{-3pt}
      \begin{itemize}[leftmargin=10pt]
      \item[--]
        $K = \bigcup_{K_i \in \psi(K,D)}
          K_i$ $\hfill\Box$
      \end{itemize}

  \end{enumerate}
\end{definition}

Let $k \in K$ be a composite key.  We write $\psi_k(K, D)$ to denote
the $\psi$-\emph{partitioning of $k$} with respect to $K$ and dimension $D$,
i.e., the partition in $\psi(K, D)$ that contains key $k$.

\begin{example}
  Let $\mathsf{K}^{1..7}$ be the set of composite keys
  from Table~\ref{tab:ex2}.  The $\psi$-partitioning of
  selected sets of keys in dimension $P$ or $V$ is as follows:
  \begin{itemize}\itemsep=0pt
    \item $\psi(\mathsf{K}^{1..7},V)=\{\mathsf{K}^{2,5,6,7},
      \mathsf{K}^{1}, \mathsf{K}^{3,4}\}$
    \item $\psi(\mathsf{K}^{2,5,6,7},P)=\{\mathsf{K}^{2}, \mathsf{K}^{5,6,7}\}$
    \item $\psi(\mathsf{K}^{5,6,7},V)=\{\mathsf{K}^{5},
      \mathsf{K}^{6}, \mathsf{K}^{7}\}$
    \item $\psi(\mathsf{K}^{6},V) = \psi(\mathsf{K}^{6},P) =
      \{\mathsf{K}^{6}\}$
  \end{itemize}

  The $\psi$-partitioning of key $\mathsf{k}_6$ with respect to sets of keys and
  dimensions is as follows:
  \begin{itemize}\itemsep=0pt
    \item $\psi_{\mathsf{k}_6}(\mathsf{K}^{1..7}, V) = \mathsf{K}^{2,5,6,7}$
    \item $\psi_{\mathsf{k}_6}(\mathsf{K}^{6}, V) =
      \psi_{\mathsf{k}_6}(\mathsf{K}^{6}, P) = \mathsf{K}^{6}$.
    $\hfill\Box$
\end{itemize}
\end{example}


The crucial property of our partitioning is that the position of the
discriminative byte for dimension $D$ increases if we $\psi$-partition
$K$ in $D$. \rev{This \emph{monotonicity} property of the
$\psi$-partitioning holds since every partition is built based on the
discriminative byte and to partition an existing partition we need a
discriminative byte that will be positioned further down in the
byte-string.}  For the alternate dimension $\overline{D}$, i.e.,
$\overline{D} = P$ if $D = V$ and $\overline{D} = V$ if $D = P$, the
position of the discriminative byte remains unchanged or may increase.

\begin{lemma}(\emph{Monotonicity of Discriminative Bytes})
  \label{lemma:monotonicity}
  Let $K_i \in \psi(K,D)$ be one of the partitions
  of $K$ after partitioning in dimension $D$.  In dimension
  $D$, the position of the discriminative byte in $K_i$ is
  strictly greater than in $K$ while, in dimension
  $\overline{D}$, the discriminative byte is equal or greater than in
  $K$, i.e.,
  \begin{align*}
  \begin{split}
    K_i &\in \psi(K,D) \land K_i \subset K
    \Rightarrow\\
    & \textsf{dsc}(K_i, D) > \textsf{dsc}(K, D)
    \land \textsf{dsc}(K_i, \overline{D}) \ge
    \textsf{dsc}(K, \overline{D})
  \end{split}
  \end{align*}
\end{lemma}

\begin{proof}
  The first line states that $K_i \subset K$ is one of the
  partitions of $K$.  From Definition~\ref{def:partitioning} it
  follows that the value $k.D[\textsf{dsc}(K,D)]$ is the same for
  every key $k \in K_i$.  From Definition~\ref{def:discbyte} it
  follows that $\textsf{dsc}(K_i,D) \neq \textsf{dsc}(K,D)$.  By
  removing one or more keys from $K$ to get $K_i$, the keys in $K_i$
  will become more similar compared to those in $K$. That means, it is
  not possible for the keys in $K_i$ to differ in a position $g <
  \textsf{dsc}(K,D)$.  Consequently, $\textsf{dsc}(K_i,D) \nless
  \textsf{dsc}(K,D)$ for any dimension $D$ (so this also holds for
  $\overline{D}$: $\textsf{dsc}(K_i,\overline{D}) \nless
  \textsf{dsc}(K,\overline{D})$).  Thus $\textsf{dsc}(K_i,D) >
  \textsf{dsc}(K,D)$ and $\textsf{dsc}(K_i,\overline{D}) \ge
  \textsf{dsc}(K,\overline{D})$.
\end{proof}

\begin{example}
  The discriminative path byte of $\mathsf{K}^{1..7}$ is $13$ while the
  discriminative value byte of $\mathsf{K}^{1..7}$ is $2$ as shown in
  Table~\ref{tab:discbyte}.  For partition $\mathsf{K}^{2,5,6,7}$, which is
  obtained by partitioning $\mathsf{K}^{1..7}$ in the value dimension, the
  discriminative path byte is $14$ while the discriminative value byte
  is $3$.  The positions of both discriminative bytes have increased.
  For partition $\mathsf{K}^{5,6,7}$, which is obtained by partitioning
  $\mathsf{K}^{2,5,6,7}$ in the path dimension, the discriminative path byte is
  $16$ while the discriminative value byte is $3$.  The position of
  the discriminative path byte has increased while the position of the
  discriminative value byte has not changed.
\end{example}

When computing the dynamic interleaving of a composite key $k \in K$
we recursively $\psi$-partition $K$ while alternating between
dimension $V$ and $P$.  This yields a partitioning sequence
$(K_1, D_1), \ldots, (K_n, D_n)$ for key $k$ with
$K_1 \supset K_2 \supset \dots \supset K_n$. We start with $K_1 = K$
and $D_1 = V$.  Next, $K_2 = \psi_k(K_1, V)$ and
$D_2 = \overline{D}_1 = P$.  We continue with the general scheme
$K_{i+1} = \psi_k(K_i, D_i)$ and $D_{i+1} = \overline{D}_i$.  This
goes on until we run out of discriminative bytes in one dimension,
which means $\psi_k(K, D) = K$.  From then on, we can only partition
in dimension $\overline{D}$.  When we run out of discriminative bytes
in this dimension as well, that is
$\psi_k(K, \overline{D}) = \psi_k(K, D) = K$, we stop. The
partitioning sequence is finite due to the monotonicity of the
$\psi$-partitioning (see Lemma~\ref{lemma:monotonicity}), which
guarantees that we make progress in every step in at least one
dimension.  Below we define a partitioning sequence.

\begin{definition}(\emph{Partitioning Sequence}) \label{def:partseq}
  The partitioning sequence
  $\rho(k,K,D) = ((K_1, D_1), \ldots, (K_n, D_n))$ of a composite key
  $k \in K$ denotes the recursive $\psi$-partitioning of the
  partitions to which $k$ belongs.  The pair $(K_i, D_i)$ denotes the
  partitioning of $K_i$ in dimension $D_i$. The final partition $K_n$
  cannot be partitioned further, hence $D_n = \bot$.  $\rho(k,K,D)$ is
  defined as follows:\footnote{Operator $\circ$ denotes concatenation, e.g.,
    $a \circ b = (a,b)$ and $a \,\circ\, (b,c) = (a,b,c)$}

  \vspace{-10pt}
  {\small \begin{align*}
    \rho(k,K,D) =
    \begin{cases}
      (K, D) \circ \rho(k,\psi_k(K,D),\overline{D})
        \hspace{0.40cm}
        \text{if }
        \psi_k(K, D) \subset K \\[1pt]
      \rho(k,K,\overline{D})
        \hspace{0.3cm}
        \text{if }
        \psi_k(K, D) = K \wedge
        \psi_k(K, \overline{D}) \subset K \\[1pt]
      (K, \bot)
        \hspace{0.8cm}
        \text{otherwise}
    \end{cases}
  \end{align*}}
\end{definition}

\begin{example} \label{ex:partseq} In the following we illustrate the
  step-by-step expansion of $\rho(\mathsf{k}_6,\mathsf{K}^{1..7},V)$
  to get $\mathsf{k}_6$'s partitioning sequence.
  \begin{align*}
    \rho(&\mathsf{k}_6,\mathsf{K}^{1..7},V) = \\
      & = (\mathsf{K}^{1..7},V)
          \circ \rho(\mathsf{k}_6, \mathsf{K}^{2,5,6,7},P) \\
      & = (\mathsf{K}^{1..7},V)
          \circ (\mathsf{K}^{2,5,6,7},P)
          \circ \rho(\mathsf{k}_6, \mathsf{K}^{5,6,7},V) \\
      & = (\mathsf{K}^{1..7},V)
          \circ (\mathsf{K}^{2,5,6,7},P)
          \circ (\mathsf{K}^{5,6,7},V)
          \circ \rho(\mathsf{k}_6, \mathsf{K}^{6},P) \\
      & = (\mathsf{K}^{1..7},V)
          \circ (\mathsf{K}^{2,5,6,7},P)
          \circ (\mathsf{K}^{5,6,7},V)
          \circ (\mathsf{K}^{6},\bot)
  \end{align*}
Notice the alternating partitioning in, respectively, $V$ and $P$.  We
only deviate from this if partitioning in one of the dimensions is not
possible.  For instance, $\mathsf{K}^{3,4}$ cannot be partitioned in
dimension $P$ and therefore we get\\[5pt]
  {\hspace*{10pt}$\begin{aligned}[b]
    \rho(\mathsf{k}_4,\mathsf{K}^{1..7},V) & =
      (\mathsf{K}^{1..7},V)
      \circ (\mathsf{K}^{3,4},V)
      \circ (\mathsf{K}^{4},\bot)
  \end{aligned} \hfill\Box$}
\end{example}

There are two key ingredients to our dynamic interleaving: the
monotonicity of discriminative bytes (Lemma~\ref{lemma:monotonicity})
and the alternating $\psi$-partitioning (Definition
\ref{def:partseq}).  The monotonicity guarantees that each time we
$\psi$-partition $K$ we advance the discriminative byte in at
least one dimension. The alternating $\psi$-partitioning ensures that
we interleave paths and values.

\subsection{Interleaving}

We determine the dynamic interleaving $I_{\text{DY}}(k,K)$ of a key $k
\in K$ via $k$'s partitioning sequence $\rho$.  For each element in
$\rho$, we generate a tuple containing two strings $s_P$ and $s_V$ and
the partitioning dimension of the element. The strings $s_P$ and $s_V$
are composed of substrings of $k.P$ and $k.V$, ranging from the
previous discriminative byte up to, but excluding, the current
discriminative byte in the respective dimension. The order of $s_P$
and $s_V$ in a tuple depends on the dimension used in the previous
step: the dimension that has been chosen for the partitioning comes
first.  Formally, this is defined as follows:

\begin{definition}(\emph{Dynamic Interleaving}) \label{def:dynint} %
  Let $k \in K$ be a composite key and let $\rho(k, K, V) = ((K_1,
  D_1),$ $\ldots,$ $(K_n, D_n))$ be the partitioning sequence of $k$.
  The dynamic interleaving $I_{\text{DY}}(k,K) = (t_1, \ldots, t_n)$
  of $k$ is a sequence of tuples $t_i$, where $t_i = (s_P, s_V, D)$ if
  $D_{i-1} = P$ and $t_i = (s_V, s_P, D)$ if $D_{i-1} = V$.  The path
  and value substrings, $s_P$ and $s_V$, and the partitioning
  dimension $D$ are determined as follows:
  \begin{align*}
    t_i.s_P & = k.P[\textsf{dsc}(K_{i-1},P),\textsf{dsc}(K_i,P)-1] \\
    t_i.s_V & = k.V[\textsf{dsc}(K_{i-1},V),\textsf{dsc}(K_i,V)-1] \\
    t_i.D & = D_i
  \end{align*}
  To correctly handle the first tuple we define
  $\textsf{dsc}(K_0,V) = 1$,
  $\textsf{dsc}(K_0,P) = 1$ and $D_0 = V$.  $\hfill\Box$
\end{definition}

\begin{example}
  We compute the tuples for the dynamic interleaving
  $I_{\text{DY}}(\mathsf{k}_6, \mathsf{K}^{1..7}) = (\mathsf{t}_1,
  \ldots, \mathsf{t}_4)$ of key $\mathsf{k}_6$ using the partitioning
  sequence $\rho(\mathsf{k}_6, \mathsf{K}^{1..7}, V) =
  ((\mathsf{K}^{1..7},V),$ $(\mathsf{K}^{2,5,6,7},P),$
  $(\mathsf{K}^{5,6,7},V),$ $(\mathsf{K}^{6},\bot))$ from Example
  \ref{ex:partseq}. The necessary discriminative path and value bytes
  can be found in Table \ref{tab:discbyte}.  Table \ref{tab:exdyn}
  shows the details of each tuple of $\mathsf{k}_6$'s dynamic interleaving with
  respect to $\mathsf{K}^{1..7}$.

\begin{table}[htbp]\centering
  \caption{Computing the dynamic interleaving $I_{\text{DY}}(\mathsf{k}_6,
  \mathsf{K}^{1..7})$.
    }
  \label{tab:exdyn}
  {\scriptsize
  \begin{tabular}{c|llc}
    {\normalsize $t$}
      & \multicolumn{1}{c}{\normalsize $s_V$}
      & \multicolumn{1}{c}{\normalsize $s_P$}
      & \multicolumn{1}{c}{\normalsize $D$} \\
    \hline
    $\mathsf{t}_1$
      & $\mathsf{k}_6.V[1,1] = \vtstr{00}$
      & $\mathsf{k}_6.P[1,12] = \ptstr{/bom/item/ca}$
      & $V$ \\
    $\mathsf{t}_2$
      & $\mathsf{k}_6.V[2,2] = \vtstr{00}$
      & $\mathsf{k}_6.P[13,13] = \ptstr{r}$
      & $P$ \\
    $\mathsf{t}_3$
      & $\mathsf{k}_6.V[3,2] = \vtstr{$\epsilon$}$
      & $\mathsf{k}_6.P[14,15] = \ptstr{/}\ptstr{b}$
      & $V$ \\
    $\mathsf{t}_4$
      & $\mathsf{k}_6.V[3,4] = \vtstr{0C}\,\vtstr{C2}$
      & $\mathsf{k}_6.P[16,20] = \ptstr{rake\$}$
      & $\bot$ \\
  \end{tabular}}
\end{table}

The final dynamic interleavings of all keys from Table \ref{tab:ex2}
are displayed in Table \ref{tab:dynints}.  We highlight in bold the
values of the discriminative bytes at which the paths and values are
interleaved, e.g., for key $\mathsf{k}_6$ these are bytes $\vtbstr{00}$,
$\ptbstr{/}$, and $\vtbstr{0C}$.

\begin{table}[htb]
\caption{The dynamic interleaving of the composite keys
in $\mathsf{K}^{1..7}$. The values
of the discriminative bytes are written in bold.}
\label{tab:dynints}
\centering
{\small
\begin{tabular}{l|l}
  \multicolumn{1}{c|}{\normalsize $k$}
    & \multicolumn{1}{c}{\normalsize Dynamic Interleaving
      $I_{\text{DY}}(k,\mathsf{K}^{1..7})$}
  \\ \hline
  $\mathsf{k}_2$
    &($(\vtstr{00}, \ptstr{/b\textcolor{black}{$\ldots$}a}, V)$,\,%
      $(\vtbstr{00}, \ptstr{r}, P)$,\,%
      $(\ptbstr{a}\ptstr{biner\$}, \vtstr{00\,F1}, \bot)$)
  \\
  $\mathsf{k}_7$
    &($(\vtstr{00}, \ptstr{/b\textcolor{black}{$\ldots$}a},V)$,\,%
      $(\vtbstr{00}, \ptstr{r},P)$,\,%
      $(\ptbstr{/}\ptstr{b}, \vtstr{$\epsilon$},V)$,\,%
      $(\vtbstr{0A}\, \vtstr{8C}, \ptstr{umper\$}, \bot)$)
  \\
  $\mathsf{k}_5$
    &($(\vtstr{00}, \ptstr{/b\textcolor{black}{$\ldots$}a}, V)$,\,%
      $(\vtbstr{00}, \ptstr{r}, P)$,\,%
      $(\ptbstr{/}\ptstr{b}, \vtstr{$\epsilon$}, V)$,\,%
      $(\vtbstr{0B}\,\vtstr{4A}, \ptstr{elt\$}, \bot)$)
  \\
  $\mathsf{k}_6$
    &($(\vtstr{00}, \ptstr{/b\textcolor{black}{$\ldots$}a}, V)$,\,%
      $(\vtbstr{00}, \ptstr{r}, P)$,\,%
      $(\ptbstr{/}\ptstr{b}, \vtstr{$\epsilon$},V)$,\,%
      $(\vtbstr{0C}\,\vtstr{C2}, \ptstr{rake\$}, \bot)$)
  \\
  $\mathsf{k}_1$
    &($(\vtstr{00}, \ptstr{/b\textcolor{black}{$\ldots$}a}, V)$,\,%
      $(\vtbstr{01}\,\vtstr{0E\,50}, \ptstr{noe\$}, \bot)$)
  \\
  $\mathsf{k}_3$
    &($(\vtstr{00}, \ptstr{/b\textcolor{black}{$\ldots$}a}, V)$,\,%
      $(\vtbstr{03}\,\vtstr{D3}, \ptstr{r/battery\$}, V)$,
      $(\vtbstr{5A}, \ptstr{$\epsilon$}, \bot)$)
  \\
  $\mathsf{k}_4$
    &($(\vtstr{00}, \ptstr{/b\textcolor{black}{$\ldots$}a}, V)$,\,%
      $(\vtbstr{03}\,\vtstr{D3}, \ptstr{r/battery\$}, V)$,
      $(\vtbstr{B0}, \ptstr{$\epsilon$}, \bot)$)
  \\
\end{tabular}}
\end{table}

\end{example}

\subsection{Efficiency of Interleavings}
\label{sec:costmodel}

We introduce a cost model to measure the efficiency of different
interleaving schemes.  We assume that the interleaved keys are
arranged in a tree-like search structure. Each node
represents the partitioning of the composite keys by either path or
value, and the node branches for each different value of a
discriminative path or value byte.  \rev{We simplify the cost model by
  assuming that the search structure is a complete tree with fanout
  $o$ where every root-to-leaf path contains $h$ edges ($h$ is called
  the height).} Further, we assume that all nodes on one level
represent a partitioning in the same dimension $\phi_i$ and we use a
vector $\phi(\phi_1, \ldots, \phi_h)$ to specify the partitioning
dimension on each level.  Figure~\ref{fig:searchscheme} visualizes
this scheme.

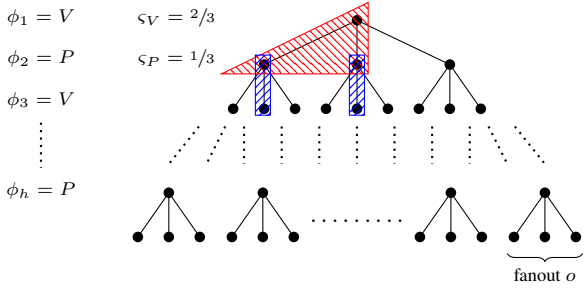
\begin{figure}[htb]
\centering
\begin{tikzpicture}
  \node at (0,0) {
  \begin{forest}
    casindex,
    for tree={
      parent anchor=north,
      l=0mm,
      s sep+=5pt,
    },
    [{},name=root,
      [{},bbnode,
        [{},bbnode,
        ]
        [{},bbnode,
        ]
        [{},bbnode,
        ]
      ]
      [{},bbnode,
        [{},bbnode,
        ]
        [{},bbnode,
        ]
        [{},bbnode,
        ]
      ]
      [{},bbnode,
        [{},bbnode,
        ]
        [{},bbnode,
        ]
        [{},bbnode,
        ]
      ]
    ]
    \fill[black] (root.parent anchor) circle[radius=2pt];
  \end{forest}
  };
  \foreach \x in {-2.5,-1.25,1.25,2.5} {
    \node (leaf-\x) at (\x,-2) {
    \begin{forest}
      searchindex,
      [{},name=root,
        []
        []
        []
      ]
      \fill[black] (root.parent anchor) circle[radius=2pt];
    \end{forest}
    };
  }
  %
  %
  \node at (0,-2) {$\ldots\ldots\ldots$};
  \draw[thick,dotted] (-2.1,-0.75)  -- (-2.5,-1.25);
  \draw[thick,dotted] (-1.75,-0.75) -- (-2,-1.25);
  \draw[thick,dotted] (+1.75,-0.75) -- (+2,-1.25);
  \draw[thick,dotted] (+2.1,-0.75)  -- (+2.5,-1.25);
  \foreach \x in {-1.5,-1.0,-0.5,0,0.5,1,1.5} {
    \draw[thick,dotted] (\x,-0.75) -- (\x,-1.25);
  }
  \draw[decorate,decoration={brace,amplitude=3pt,raise=4pt,mirror},yshift=0pt]
    (2.0,-2.3) -- (3.0,-2.3)
    node[midway,yshift=-12pt,font={\scriptsize}]
    {fanout $o$};
  \node[font=\scriptsize] at (-4.2,0.7)  {$\phi_1 = V$};
  \node[font=\scriptsize] at (-4.2,0.15) {$\phi_2 = P$};
  \node[font=\scriptsize] at (-4.2,-0.4) {$\phi_3 = V$};
  \node[font=\scriptsize] at (-4.2,-1.6) {$\phi_h = P$};
  \draw[thick,dotted] (-4.2,-0.7)  -- (-4.2,-1.3);
  %
  %
  \draw[draw=red,pattern=north west lines, pattern color=red]
    (0.15,0.9) -- (-1.8,-0.05) -- (0.15,-0.05) -- cycle;
  \node[font=\scriptsize] at (-2.4,0.7) {$\varsigma_V = \sfrac{2}{3}$};
  %
  \fill[draw=blue,pattern=north east lines, pattern color=blue]
    (-1.35,+0.2) -- (-1.35,-0.6) -- (-1.15,-0.6) -- (-1.15,0.2) -- cycle;
  \fill[draw=blue,pattern=north east lines, pattern color=blue]
    (0.1,+0.2) -- (0.1,-0.6) -- (-0.1,-0.6) -- (-0.1,0.2) -- cycle;
  \node[font=\scriptsize] at (-2.4,0.15) {$\varsigma_P = \sfrac{1}{3}$};
\end{tikzpicture}
\caption{The search structure in our cost model is a complete tree of
height $h$ and fanout $o$.}
\label{fig:searchscheme}
\end{figure}

A search starts at the root and traverses the data structure to
determine the answer set. In the case of range queries, more than one
branch must be followed.  A search follows a fraction of the outgoing
branches $o$ originating at this node. We call this the selectivity of
a node (or just selectivity).  We assume that every path node has a
selectivity of $\rev{\varsigma_P}$ and every value node has the
selectivity of $\rev{\varsigma_V}$.  The cost $\rev{\widehat{C}}$ of a
search, measured in the number of visited nodes during the search, is
as follows: \rev{\begin{equation*}
    \widehat{C}(o,h,\phi,\varsigma_P,\varsigma_V) = 1 + \sum_{l=1}^{h}
    \prod_{i=1}^{l} (o \cdot \varsigma_{\phi_i})
\end{equation*}}

If a workload is known upfront, a system can optimize indexes to
support specific queries.  Our goal is an access method that can deal
with a wide range of queries in a dynamic environment in a robust way,
i.e., avoiding a bad performance for any particular query type.  This
is motivated by the fact that modern data analytics utilizes a large
number of ad-hoc queries to do exploratory analysis. For example, in
the context of building a robust partitioning scheme for ad-hoc query
workloads, Shanbhag et al.~\cite{Shanbhag17} found that after
analyzing the first 80\% of real-world workload traces the remaining
20\% still contained 57\% completely new queries.

Even though robustness of query processing performance has received
considerable interest, there is a lack of unified metrics in this
area~\cite{Graefe11,Graefe12}.  Our goal is a good performance for
queries with differing selectivities for path and value predicates.
Towards this goal we define the notion of \emph{complementary
  queries}.

\begin{definition}(\emph{Complementary Query}) %
  Given a query $Q$ with path selectivity $\rev{\varsigma_P}$ and value
  selectivity $\rev{\varsigma_V}$, there is a \emph{complementary query} $Q'$
  with path selectivity $\rev{\varsigma'_P} = \rev{\varsigma_V}$ and value selectivity
  $\rev{\varsigma'_V} = \rev{\varsigma_P}$
\end{definition}

State-of-the-art CAS-indexes favor either path or value predicates.
As a result they show a very good performance for one type of query
but run into problems for the complementary query.

\begin{definition}(\emph{Robustness}) \label{def:robustness}
  A CAS-index is \emph{robust} if it optimizes the average performance
  when evaluating a query $Q$ and its complementary query $Q'$.
\end{definition}

\begin{example} \label{ex:cost}
  Figure~\ref{fig:interleavecost}a shows the costs for a query $Q$ and
  its complementary query $Q'$ for different interleavings in terms of
  the number of visited nodes during the search.  We assume parameters
  $o=10$ and $h=12$ for the search structure.  In our dynamic
  interleaving $I_{\text{DY}}$ the discriminative bytes are perfectly
  alternating.  $I_{\text{PV}}$ stands for path-value concatenation
  with $\phi_i = P$ for $1 \leq i \leq 6$ and $\phi_i = V$ for $7 \leq
  i \leq 12$.  $I_{\text{VP}}$ is a value-path concatenation (with an
  inverse $\phi$ compared to $I_{\text{PV}}$).  We also consider two
  additional permutations: $I_1$ uses a vector $\phi =
  (V,V,V,V,P,V,P,V,P,P,P,P)$ and $I_2$ one equal to
  $(V,V,V,P,P,V,P,V,V,P,P,P)$.  They resemble, e.g., the byte-wise
  interleaving that usually exhibits irregular alternation patterns
  with a clustering of, respectively, discriminative path and value
  bytes.  Figure \ref{fig:interleavecost}b shows the average costs and
  the standard deviation.  The numbers demonstrate the robustness of
  our dynamic interleaving: it clearly shows the best performance both
  in terms of average costs and lowest standard deviation.
\end{example}

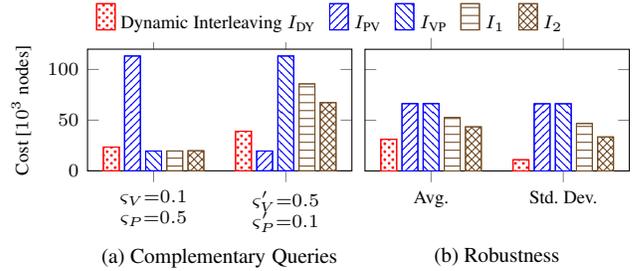
\begin{figure}[htb]
\begin{tikzpicture}
  \begin{groupplot}[
    height=32mm,
    width=145pt,
    ticklabel style={font=\scriptsize},
    xlabel near ticks,
    ylabel near ticks,
    ylabel/.append style={font=\scriptsize},
    scaled y ticks=base 10:-3,
    ytick scale label code/.code={},
    symbolic x coords={k1,k2},
    xtick=data,
    xticklabel style={
      align=center,
    },
    group style={
      group size=2 by 1,
      horizontal sep=5pt,
      yticklabels at=edge left,
      ylabels at=edge left,
    },
    ybar,
    enlarge x limits=0.5,
    bar cycle list,
  ]
  \nextgroupplot[
    bar width=6pt,
    ylabel={Cost [$10^3$ nodes]},
    xticklabels={
      {$\rev{\varsigma_V}{=}0.1$\\$\rev{\varsigma_P}{=}0.5$},
      {$\rev{\varsigma'_V}{=}0.5$\\$\rev{\varsigma'_P}{=}0.1$}
    },
    legend style={
      at={(0.0,1.05)},
      anchor=south west,
      legend columns=-1,
      draw=none,
      /tikz/every even column/.append style={column sep=3pt},
    },
    xlabel={\small (a) Complementary Queries},
    ymax=120000,
    ymin=0,
  ]
  \addplot[idy]  coordinates {(k1, 23436) (k2, 39060)};
  \addplot[ipv]  coordinates {(k1,113280) (k2, 19536)};
  \addplot[ivp]  coordinates {(k1, 19536) (k2,113280)};
  \addplot[ione] coordinates {(k1, 19564) (k2, 85780)};
  \addplot[itwo] coordinates {(k1, 19808) (k2, 67280)};
  \legend{Dynamic Interleaving $I_\text{DY}$, $I_{\text{PV}}$,
    $I_{\text{VP}}$, $I_1$, $I_2$}
  \nextgroupplot[
    bar width=6pt,
    xticklabels={Avg., Std.~Dev.},
    xlabel={\small (b) Robustness},
    x label style={yshift=-9pt},
    ymax=120000,
    ymin=0,
  ]
  \addplot[idy]  coordinates {(k1, 31248) (k2, 11047)};
  \addplot[ipv]  coordinates {(k1, 66408) (k2, 66287)};
  \addplot[ivp]  coordinates {(k1, 66408) (k2, 66287)};
  \addplot[ione] coordinates {(k1, 52672) (k2, 46821)};
  \addplot[itwo] coordinates {(k1, 43544) (k2, 33567)};
  \end{groupplot}
\end{tikzpicture}
\caption{Dynamic interleaving has a robust query performance.}
\label{fig:interleavecost}
\end{figure}

In the previous example we showed empirically that a perfectly
alternating interleaving exhibits the best overall performance when
evaluating complementary queries. In addition to this, we can prove
that this is always the case.

\begin{theorem}\label{theo:avg}
  Consider a query $Q$ with selectivities $\rev{\varsigma_P}$ and
  $\rev{\varsigma_V}$
  and its complementary query $Q'$ with selectivities
  $\rev{\varsigma'_P} = \rev{\varsigma_V}$ and $\rev{\varsigma'_V} =
  \rev{\varsigma_P}$.  There is no
  interleaving that on average performs better than the dynamic
  interleaving that has a perfectly alternating vector
  $\phi_{\text{DY}}$, i.e.,
  $\forall \phi:
  \rev{\widehat{C}}(o,h,\phi_{\text{DY}},\rev{\varsigma_P},\rev{\varsigma_V}) +
  \rev{\widehat{C}}(o,h,\phi_{\text{DY}},\rev{\varsigma'_P},\rev{\varsigma'_V}) \leq
  \rev{\widehat{C}}(o,h,\phi,\rev{\varsigma_P},\rev{\varsigma_V}) +
  \rev{\widehat{C}}(o,h,\phi,\rev{\varsigma'_P},\rev{\varsigma'_V})$.%
\end{theorem}

\begin{proof}
  \rev{We begin with a brief outline of the proof. We show for a level
  $l$ that the costs of query $Q$ and complementary query $Q'$ on
  level $l$ is smallest with the dynamic interleaving.  That is, for a
  level $l$ we show that $\prod_{i=1}^{l} (o \cdot \varsigma_{\phi_i}) +
  \prod_{i=1}^{l} (o \cdot \varsigma'_{\phi_i})$ is smallest with the
  vector $\phi_{\text{DY}} = (V,P,V,P,\ldots)$ of our dynamic
  interleaving.  Since this holds for any level $l$, it also holds for
  the sum of costs over all levels $l$, $1 \leq l \leq h$, and this
  proves the theorem.  }

  We only look at search trees with a height $h \geq 2$, as for $h=1$
  we do not actually have an interleaving (and the costs are all the
  same). W.l.o.g., we assume that the first level of the search tree
  always starts with a discriminative value byte, i.e., $\phi_1 = V$.
  Let us look at the cost for one specific level $l$ for query $Q$ and
  its complementary query $Q'$. We distinguish two cases: $l$ is even
  or $l$ is odd.

  \textbf{$l$ is even:} The cost for a perfectly
  alternating interleaving for $Q$ for level $l$ is equal to $o^l
  (\varsigma_V \cdot \varsigma_P \dots \varsigma_V \cdot
  \varsigma_P)$, while the cost for $Q'$ is equal to $o^l
  (\varsigma_V' \cdot \varsigma_P' \dots \varsigma_V' \cdot
  \varsigma_P')$, which is equal to $o^l (\varsigma_P \cdot
  \varsigma_V \dots \varsigma_P \cdot \varsigma_V)$. This is the same
  cost as for $Q$, so adding the two costs gives us $2 o^l
  \varsigma_V^{\sfrac{l}{2}} \varsigma_P^{\sfrac{l}{2}}$

  \rev{For a non-perfectly alternating interleaving with the same
  number of $\varsigma_V$ and $\varsigma_P$ multiplicands up to level
  $l$ we have the same cost as for our dynamic interleaving, i.e., $2
  o^l \varsigma_V^{\sfrac{l}{2}} \varsigma_P^{\sfrac{l}{2}}$. Now let
  us assume that the number of $\varsigma_V$ and $\varsigma_P$
  multiplicands is different for level $l$ (there must be at least one
  such level $l$).} Assume that for $Q$ we have $r$ multiplicands of
  type $\varsigma_V$ and $s$ multiplicands of type $\varsigma_P$, with
  $r+s = l$ and, w.l.o.g., $r>s$. This gives us $o^l \varsigma_V^s
  \varsigma_P^s \varsigma_V^{r-s} + o^l \varsigma_V^s \varsigma_P^s
  \varsigma_P^{r-s} = o^l \varsigma_V^s \varsigma_P^s
  (\varsigma_V^{r-s} + \varsigma_P^{r-s})$ for the cost.

  We have to show that $2 o^l \varsigma_V^{\sfrac{l}{2}}
  \varsigma_P^{\sfrac{l}{2}} \leq o^l \varsigma_V^s \varsigma_P^s
  (\varsigma_V^{r-s} + \varsigma_P^{r-s})$. As all values are greater
  than zero, this is equivalent to $2 \varsigma_V^{\sfrac{l}{2}-s}
  \varsigma_P^{\sfrac{l}{2}-s} \leq \varsigma_V^{r-s} +
  \varsigma_P^{r-s}$. The right-hand side can be reformulated:
  $\varsigma_V^{r-s} + \varsigma_P^{r-s} = \varsigma_V^{l-2s} +
  \varsigma_P^{l-2s} = \varsigma_V^{\sfrac{l}{2}-s}
  \varsigma_V^{\sfrac{l}{2}-s} + \varsigma_P^{\sfrac{l}{2}-s}
  \varsigma_P^{\sfrac{l}{2}-s}$.  Setting $a =
  \varsigma_V^{\sfrac{l}{2}-s}$ and $b =
  \varsigma_P^{\sfrac{l}{2}-s}$, this boils down to showing $2ab \leq
  a^2 + b^2 \Leftrightarrow 0 \leq (a - b)^2$, which is always true.

  \textbf{$l$ is odd:} W.l.o.g. we assume that for
  computing the cost for a perfectly alternating interleaving for $Q$,
  there are $\lceil \sfrac{l}{2} \rceil$ multiplicands of type
  $\varsigma_V$ and $\lfloor \sfrac{l}{2} \rfloor$ multiplicands of
  type $\varsigma_P$. This results in $o^l \varsigma_V^{\lfloor
  \sfrac{l}{2} \rfloor} \varsigma_P^{\lfloor \sfrac{l}{2} \rfloor}
  (\varsigma_V + \varsigma_P)$ for the sum of costs for $Q$ and $Q'$.

  For a non-perfectly alternating interleaving, we again have $o^l
  \varsigma_V^s \varsigma_P^s (\varsigma_V^{r-s} + \varsigma_P^{r-s})$
  with $r+s = l$ and $r>s$, which can be reformulated to $o^l
  \varsigma_V^s \varsigma_P^s (\varsigma_V^{\lfloor \sfrac{l}{2}
  \rfloor - s} \varsigma_V^{\lfloor \sfrac{l}{2} \rfloor - s}
  \varsigma_V + \varsigma_P^{\lfloor \sfrac{l}{2} \rfloor - s}
  \varsigma_P^{\lfloor \sfrac{l}{2} \rfloor - s} \varsigma_P)$.

  What is left to prove is $o^l \varsigma_V^{\lfloor \sfrac{l}{2}
  \rfloor} \varsigma_P^{\lfloor \sfrac{l}{2} \rfloor} (\varsigma_V +
  \varsigma_P) \leq o^l \varsigma_V^s \varsigma_P^s
  (\varsigma_V^{\lfloor \sfrac{l}{2} \rfloor - s} \varsigma_V^{\lfloor
  \sfrac{l}{2} \rfloor - s} \varsigma_V + \varsigma_P^{\lfloor
  \sfrac{l}{2} \rfloor - s} \varsigma_P^{\lfloor \sfrac{l}{2} \rfloor
  - s} \varsigma_P)$, which is equivalent to $\varsigma_V^{\lfloor
  \sfrac{l}{2} \rfloor - s} \varsigma_P^{\lfloor \sfrac{l}{2} \rfloor
  - s} (\varsigma_V + \varsigma_P) \leq \varsigma_V^{\lfloor
  \sfrac{l}{2} \rfloor - s} \varsigma_V^{\lfloor \sfrac{l}{2} \rfloor
  - s} \varsigma_V + \varsigma_P^{\lfloor \sfrac{l}{2} \rfloor - s}
  \varsigma_P^{\lfloor \sfrac{l}{2} \rfloor - s} \varsigma_P$.
  Substituting $a = \varsigma_V$, $b = \varsigma_P$, and $x = {\lfloor
  \sfrac{l}{2} \rfloor - s}$, this means showing that $a^x b^x (a+b)
  \leq a^{2x + 1} + b^{2x + 1} \Leftrightarrow 0 \leq a^{2x + 1} +
  b^{2x + 1} - a^x b^x (a+b)$. Factorizing this polynomial gives us
  $(a^x - b^x)(a^{x+1} - b^{x+1})$ or $(b^x - a^x)(b^{x+1} -
  a^{x+1})$. We look at $(a^x - b^x)(a^{x+1} - b^{x+1})$, the argument
  for the other factorization follows along the same lines.  This term
  can only become negative if one factor is negative and the other is
  positive. Let us first look at the case $a < b$: since $0 \leq a,b
  \leq 1$, we can immediately follow that $a^x < b^x$ and $a^{x+1} <
  b^{x+1}$, i.e., both factors are negative. Analogously, from $a > b$
  (and $0 \leq a,b \leq 1$) immediately follows $a^x > b^x$ and
  $a^{x+1} > b^{x+1}$, i.e., both factors are positive.
\end{proof}


\rev{
  Note that in practice the search structure is not a complete tree
  and the fraction $\varsigma_P$ and $\varsigma_V$ of children that
  are traversed at each node is not constant.  In
  Section~\ref{sec:expcostmodel} we experimentally evaluate the cost
  model on real-world datasets.  We show that the estimated cost and
  the true cost of a query are off by a factor of two, on average.
  This is a good estimate for the cost of a query.  }



\section{RCAS Index}
\label{sec:rcas}

We propose the Robust Content-And-Structure (RCAS) index to
efficiently query the content and structure of hierarchical data.
The RCAS index uses our dynamic interleaving to integrate the paths
and values of composite keys in a trie-based index.

\subsection{Trie-Based Structure of RCAS}

The RCAS index is a trie data-structure that efficiently supports CAS
queries with range and prefix searches.  Each node $n$ in the RCAS
index includes a dimension $n.D$, path substring $n.s_P$, and value
substring $n.s_V$ that correspond to the fields $t.D$, $t.s_P$ and
$t.s_V$ in the dynamic interleaving of a key (see Definition
\ref{def:dynint}). The substrings $n.s_P$ and $n.s_V$ are
variable-length strings. Dimension $n.D$ is $P$ or $V$ for inner nodes
and $\bot$ for leaf nodes. Leaf nodes additionally store a set of
references $r_i$ to nodes in the database, denoted $n.\text{refs}$.
Each dynamically interleaved key corresponds to a root-to-leaf path in
the RCAS index.

\begin{definition}(\emph{RCAS Index}) \label{def:rcas} %
  Let $K$ be a set of composite keys and let $R$ be a tree.
  Tree $R$ is the RCAS index for $K$ iff the following
  conditions are satisfied.
  \begin{enumerate}
  \item $I_\text{DY}(k,K) = (t_1,\ldots,t_m)$ is the dynamic
    interleaving of a key $k \in K$ iff there is a
    root-to-leaf path $(n_1, \ldots, n_m)$ in $R$ such that
    $t_i.s_P = n_i.s_P$, $t_i.s_V = n_i.s_V$, and $t_i.D = n_i.D$ for
    $1 \leq i \leq m$.
  \item $R$ does not include duplicate siblings, i.e., no two sibling
    nodes $n$ and $n'$, $n \neq n'$, in $R$ have the same values for
    $s_P$, $s_V$, and $D$, respectively.
  \end{enumerate}
\end{definition}

\begin{example}
  Figure \ref{fig:solution} shows the RCAS index for the composite
  keys $\mathsf{K}^{1..7}$.  We use blue and red colors for bytes
  from the path and value, respectively.  The discriminative bytes are
  highlighted in bold.  The dynamic interleaving
  $I_{\text{DY}}(\mathsf{k}_6,\mathsf{K}^{1..7}) = (\mathsf{t}_1,
  \mathsf{t}_2, \mathsf{t}_3, \mathsf{t}_4)$ from Table
  \ref{tab:dynints} is mapped to the path $(\mathsf{n}_1,
  \mathsf{n}_2, \mathsf{n}_4, \mathsf{n}_7)$ in the RCAS index. For
  key $\mathsf{k}_2$, the first two tuples of
  $I_{\text{DY}}(\mathsf{k}_2,\mathsf{K}^{1..7})$ are also mapped to
  $\mathsf{n}_1$ and $\mathsf{n}_2$, while the third tuple is mapped
  to $\mathsf{n}_3$.
\end{example}

\begin{figure}[htbp]\centering
\scalebox{.75}{
\begin{forest}
  casindex,
  for tree={
    s sep-=1mm,
  },
  [{$\mathsf{n}_{1}$}\\{(\vtstr{00},\pstr{/bom/item/ca},$V$)},name=root
    [{$\mathsf{n}_{2}$}\\{(\vtbstr{00},\pstr{r},$P$)},pvnode,
      [{{$\mathsf{n}_{3}$}\\(\ptbstr{a}\pstr{biner\$},\vtstr{00\,F1},$\bot$)\\$\{\mathsf{r}_{2}\}$},lpnode,
      ]
      [{$\mathsf{n}_{4}$}\\{(\ptbstr{/}\pstr{b},\vtstr{$\epsilon$},$V$)},vpnode,
        [{{$\mathsf{n}_{5}$}\\(\vtstr{{\bf \ttfamily 0A}\,8C},\pstr{umper\$},$\bot$)\\$\{\mathsf{r}_{7}\}$},lvnode,
        ]
        [{{$\mathsf{n}_{6}$}\\(\vtstr{{\bf \ttfamily 0B}\,4A},\pstr{elt\$},$\bot$)\\$\{\mathsf{r}_{5}\}$},lvnode,
        ]
        [{{$\mathsf{n}_{7}$}\\(\vtstr{{\bf \ttfamily 0C}\,C2},\pstr{rake\$},$\bot$)\\$\{\mathsf{r}_{6}\}$},lvnode,
        ]
      ]
    ]
    [{{$\mathsf{n}_{8}$}\\(\vtstr{{\bf \ttfamily 01}\,0E\,50},\pstr{noe\$},$\bot$)\\$\{\mathsf{r}_{1}\}$},lvnode,
    ]
    [{$\mathsf{n}_{9}$}\\{(\vtstr{{\bf \ttfamily 03}\,D3},\pstr{r/battery\$},$V$)},vvnode,
      [{{$\mathsf{n}_{10}$}\\(\vtbstr{5A},\pstr{$\epsilon$},$\bot$)\\$\{\mathsf{r}_{3},\mathsf{r}'_3\}$},lvnode,
      ]
      [{{$\mathsf{n}_{11}$}\\(\vtbstr{B0},\pstr{$\epsilon$},$\bot$)\\$\{\mathsf{r}_{4}\}$},lvnode,
      ]
    ]
  ]
  \fill[black] (root.parent anchor) circle[radius=2pt];
  \node at ([yshift=5pt]root.north) {\ttfamily };
\end{forest}
}
\caption{The RCAS index for the composite keys $\mathsf{K}^{1..7}$.}
\label{fig:solution}
\end{figure}
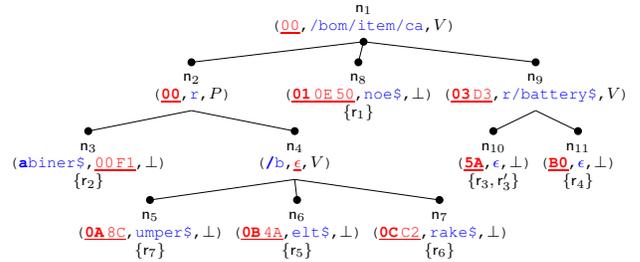

\subsection{Physical Node Layout}
\label{sec:nodelayout}

Figure \ref{fig:node} shows the physical structure of an inner node.
The \pcode{header} field contains meta information, such as the number
of children.  Fields $s_P$ and $s_V$ (explained above) are implemented
as variable-length byte vectors (C++'s \pcode{std::vector<uint8_t>}).
Dimension $D$ ($P$ or $V$, or $\bot$ if the node is a leaf) is the
dimension in which the node partitions the data.  The remaining space
of an inner node (gray-colored in Figure \ref{fig:node}) is reserved
for child pointers.  Since $\psi$ partitions at the granularity of
bytes, each node can have at most 256 children, one for each possible
value of a discriminative byte from \pcode{0x00} to \pcode{0xFF} (or their
corresponding ASCII characters in Figure \ref{fig:solution}).  For
each possible value $b$ there is a pointer to the subtree whose keys
all have value $b$ for the discriminative byte of dimension $D$.
Typically, many of the 256 pointers are \pcode{NULL}.  Therefore, we
implement our trie as an Adaptive Radix Tree (ART)~\cite{VL13}, which
defines four node types with a physical fanout of 4, 16, 48, and 256
child pointers, respectively.  Nodes are resized to adapt to the
actual fanout of a node.  Figure \ref{fig:node} illustrates the node
type with an array of 256 child pointers.  For the remaining node
types we refer to Leis et al.\ \cite{VL13}.

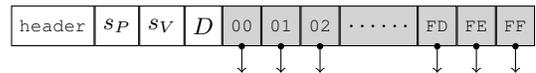
\begin{figure}[htb]
\centering
\begin{tikzpicture}[
  box/.style={
    draw,
    minimum height=15pt,
    anchor=west,
  },
  ptr/.style={
    fill=colxe,
  },
]
  \node[box] (n0) {\pcodes{header}};
  \node[box] (n1) at (n0.east) {$s_P$};
  \node[box] (n2) at (n1.east) {$s_V$};
  \node[box] (n3) at (n2.east) {$D$};
  \node[box,ptr] (n00) at (n3.east) {\pcodes{00}};
  \node[box,ptr] (n01) at (n00.east) {\pcodes{01}};
  \node[box,ptr] (n02) at (n01.east) {\pcodes{02}};
  \node[box,ptr] (nXX) at (n02.east) {$\ldots\ldots$};
  \node[box,ptr] (nFD) at (nXX.east) {\pcodes{FD}};
  \node[box,ptr] (nFE) at (nFD.east) {\pcodes{FE}};
  \node[box,ptr] (nFF) at (nFE.east) {\pcodes{FF}};
  \draw[fill=black] (n00.south) circle (1pt) edge[->] ([yshift=-10pt] n00.south) ;
  \draw[fill=black] (n01.south) circle (1pt) edge[->] ([yshift=-10pt] n01.south) ;
  \draw[fill=black] (n02.south) circle (1pt) edge[->] ([yshift=-10pt] n02.south) ;
  \draw[fill=black] (nFD.south) circle (1pt) edge[->] ([yshift=-10pt] nFD.south) ;
  \draw[fill=black] (nFE.south) circle (1pt) edge[->] ([yshift=-10pt] nFE.south) ;
  \draw[fill=black] (nFF.south) circle (1pt) edge[->] ([yshift=-10pt] nFF.south) ;
\end{tikzpicture}
\caption{Structure of an inner node with 256 pointers.}
\label{fig:node}
\end{figure}

The structure of a leaf node is similar to that shown in Figure
\ref{fig:node}, except that instead of child pointers the leaf nodes
have a variable-length vector with references to nodes in the
database.


\subsection{\rev{Bulk-Loading} RCAS}

\rev{This section gives an efficient bulk-loading algorithm for RCAS
that is linear in the number of composite keys $|K|$. It
simultaneously computes the dynamic interleaving of all keys in $K$.}
We implement a partition $K$ as a linked list of pairs $(k,r)$, where
$r$ is a reference to a database node with path $k.P$ and value $k.V$.
In our implementation the keys in $K$ need not be unique. There can be
pairs $(k,r_i)$ and $(k,r_j)$ that have the same key but have
different references $r_i \neq r_j$. This is the case if there are
different nodes in the indexed database that have the same path and
value (thus the same key).  A partitioning $M = \psi(K, D)$ is
implemented as an array of length $2^8$ with references to (possibly
empty) partitions $K$.  The array indexes \texttt{0x00} to
\texttt{0xFF} are the values of the discriminative byte.

\begin{example}
  Figure \ref{fig:structures}a shows the linked list for set
  $\mathsf{K}^{1..7}$ from our running example. Two nodes, pointed to
  by $\mathsf{r}_3$ and $\mathsf{r}'_3$, have the same key
  $\mathsf{k}_3$. They correspond to the batteries in Figure
  \ref{fig:bom} that have the same path and value.  Figure
  \ref{fig:structures}b shows the partitioning
  $\psi(\mathsf{K}^{1..7}, V)$ for our running example. Three
  partitions exist with values \vtbstr{0x00}, \vtbstr{0x01}, and
  \vtbstr{0x03} for the discriminative value byte.
\end{example}

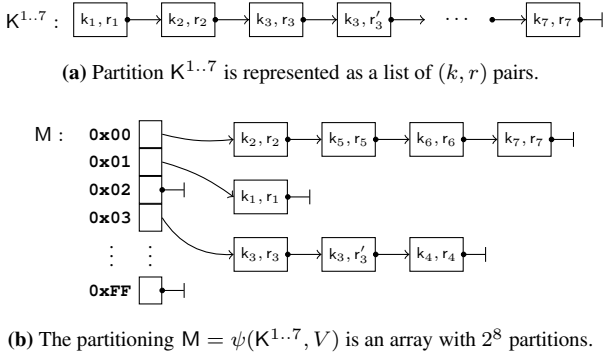
\begin{figure}[htb]
\begin{subfigure}{\linewidth}
\centering
\begin{tikzpicture}[
  scale=0.85,
  every node/.style={
    scale=0.85,
  },
  box/.style={
    draw,
    minimum height=15pt,
    anchor=west,
    font={\scriptsize},
  },
  ptr/.style={
    fill=colxe,
  },
]
  \node (label) {$\mathsf{K}^{1..7}:$};
  \node[box] (k1) at ([xshift=3pt] label.east) {$\mathsf{k}_1,\mathsf{r}_1$};
  \node[box] (k2) at ([xshift=15pt] k1.east) {$\mathsf{k}_2,\mathsf{r}_2$};
  \node[box] (k3) at ([xshift=15pt] k2.east) {$\mathsf{k}_3,\mathsf{r}_3$};
  \node[box] (k4) at ([xshift=15pt] k3.east) {$\mathsf{k}_3,\mathsf{r}'_3$};
  \node[box] (k8) at ([xshift=60pt] k4.east) {$\mathsf{k}_7,\mathsf{r}_7$};
  \node at ([xshift=30pt] k4.east) {$\cdots$};
  \draw[fill=black] (k1.east) circle (1pt) edge[->] (k2.west);
  \draw[fill=black] (k2.east) circle (1pt) edge[->] (k3.west);
  \draw[fill=black] (k3.east) circle (1pt) edge[->] (k4.west);
  \draw[fill=black] (k4.east) circle (1pt) edge[->] ([xshift=15pt] k4.east);
  \draw[fill=black] ([xshift=-15pt] k8.west) circle (1pt) edge[->] (k8.west);
  \draw[fill=black] (k8.east) circle (1pt) edge[-|] ([xshift=10pt] k8.east);
\end{tikzpicture}
\par\smallskip
\caption{Partition $\mathsf{K}^{1..7}$ is represented as a list of $(k,r)$ pairs.}
\end{subfigure}
\par\bigskip
\begin{subfigure}{\linewidth}
\centering
\begin{tikzpicture}[
  scale=0.85,
  every node/.style={
    scale=0.85,
  },
  box/.style={
    draw,
    minimum height=15pt,
    anchor=west,
    font={\scriptsize},
  },
  ptr/.style={
    fill=colxe,
  },
  slot/.style={
    draw,
    minimum height=12pt,
    minimum width=10pt,
    anchor=north,
  },
  slotidx/.style={
    anchor=east,
    font={\small \tt \bf},
  },
]
  \node[slot] (a00) at (0,0) {};
  \node[slot] (a01) at (a00.south) {};
  \node[slot] (a02) at (a01.south) {};
  \node[slot] (a03) at (a02.south) {};
  \node[slot] (a04) at ([yshift=-20pt] a03.south) {};
  %
  \node[slotidx] at (a00.west) {\texttt{0x00}};
  \node[slotidx] at (a01.west) {\texttt{0x01}};
  \node[slotidx] at (a02.west) {\texttt{0x02}};
  \node[slotidx] at (a03.west) {\texttt{0x03}};
  \node[slotidx] at (a04.west) {\texttt{0xFF}};
  %
  \node (dots) at ([yshift=-8pt] a03.south) {$\vdots$};
  \node at ([xshift=-13pt] dots.west) {$\vdots$};
  %
  \draw[fill=black] (a02.east) circle (1pt) edge[-|] ([xshift=10pt] a02.east);
  \draw[fill=black] (a04.east) circle (1pt) edge[-|] ([xshift=10pt] a04.east);
  %
  \node[box] (k2) at (1.3,-0.3) {$\mathsf{k}_2,\mathsf{r}_2$};
  \node[box] (k5) at ([xshift=15pt] k2.east) {$\mathsf{k}_5,\mathsf{r}_5$};
  \node[box] (k6) at ([xshift=15pt] k5.east) {$\mathsf{k}_6,\mathsf{r}_6$};
  \node[box] (k7) at ([xshift=15pt] k6.east) {$\mathsf{k}_7,\mathsf{r}_7$};
  \draw[fill=black] (k2.east) circle (1pt) edge[->] (k5.west);
  \draw[fill=black] (k5.east) circle (1pt) edge[->] (k6.west);
  \draw[fill=black] (k6.east) circle (1pt) edge[->] (k7.west);
  \draw[fill=black] (k7.east) circle (1pt) edge[-|] ([xshift=10pt] k7.east);
  %
  \node[box] (k1) at (1.3,-1.2) {$\mathsf{k}_1,\mathsf{r}_1$};
  \draw[fill=black] (k1.east) circle (1pt) edge[-|] ([xshift=10pt] k1.east);
  %
  \node[box] (k3) at (1.3,-2.1) {$\mathsf{k}_3,\mathsf{r}_3$};
  \node[box] (k3p) at ([xshift=15pt] k3.east) {$\mathsf{k}_3,\mathsf{r}'_3$};
  \node[box] (k4) at ([xshift=15pt] k3p.east) {$\mathsf{k}_4,\mathsf{r}_4$};
  \draw[fill=black] (k3.east) circle (1pt) edge[->] (k3p.west);
  \draw[fill=black] (k3p.east) circle (1pt) edge[->] (k4.west);
  \draw[fill=black] (k4.east) circle (1pt) edge[-|] ([xshift=10pt] k4.east);
  \draw (a00.east) edge[->,bend right=10] (k2.west);
  \draw (a01.east) edge[->,bend left=10]  (k1.west);
  \draw (a03.east) edge[->,bend right=30] (k3.west);
  %
  \node at ([xshift=-40pt] a00.west) {$\mathsf{M}:$};
\end{tikzpicture}
\par\smallskip
\caption{The partitioning $\mathsf{M} = \psi(\mathsf{K}^{1..7},V)$ is
an array with $2^8$ partitions.}
\end{subfigure}
\caption{Data structures used in Algorithm \ref{alg:bulk}.}
\label{fig:structures}
\end{figure}

Algorithm~\ref{alg:dsc} determines the discriminative byte for a
partition $K$.  Note that $\textsf{dsc\_inc}$ looks for the
discriminative byte starting from position $g$, where $g$ is a lower
bound for $\textsf{dsc}(K, D)$ as per Lemma~\ref{lemma:monotonicity}.
Also, looping through the bytes of the first key of $K$ is correct
even if there are shorter keys in $K$. Since we use prefix-free keys,
any shorter keys differ at some position, terminating the loop
correctly.


\begin{algorithm2e}[htb] \scriptsize
  Let $(k_i, r_i)$ be the first key in $K$\;
  \While{\textnormal{$g \leq \textsf{len}(k_i.D)$}}{
    \For{$(k_j,r_j) \in K$} {
      \lIf{$k_j.D[g] \neq k_i.D[g]$} {%
        \Return{$g$}%
      }
    }
    $g\text{++}$\;
  }
  \Return{$g$}
  \caption{$\textsf{dsc\_inc}(K, D, g)$}
  \label{alg:dsc}
\end{algorithm2e}

Algorithm~\ref{alg:psi} illustrates the computation of the
$\psi$-partitioning $M = \psi(K, D, g)$.  We pass the position $g$ of
the discriminative byte for dimension $D$ as an argument to $\psi$.
The discriminative byte determines the partition to which a key
belongs.

\begin{algorithm2e}[htb] \scriptsize
  Let $M$ be an array of $2^8$ empty lists\;
  \For{$(k_i,r_i) \in K$} {
    Move $(k_i,r_i)$ from partition $K$ to partition $M[k_i.D[g]]$\;
  }
  \Return{$M$}
  \caption{$\psi(K, D, g)$}
  \label{alg:psi}
\end{algorithm2e}

Algorithm~\ref{alg:bulk} recursively $\psi$-partitions $K$ and
alternates between the path and value dimensions.  In each call of
\textsf{BulkLoadRCAS} one new node $n$ is added to the index.  The
algorithm takes four parameters: (a) partition $K$, (b) dimension
$D \in \{P,V\}$ by which $K$ is partitioned, and the positions of the
previous discriminative (c) path byte $g_P$ and (d) value byte $g_V$.
For the first call, when no previous discriminative path and value
bytes exist, we set $g_P = g_V = 1$.
$\textsf{BulkLoadRCAS}(K, V, 1, 1)$ returns a pointer to the root
node of the new RCAS index.  We start by creating a new node $n$ (line
1) and determining the discriminative path and value bytes $g'_P$ and
$g'_V$ of $K$ (lines 3-4).  Lemma~\ref{lemma:monotonicity} guarantees
that the previous discriminative bytes $g_P$ and $g_V$ are valid lower
bounds for $g'_P$ and $g'_V$, respectively. Next, we determine the
current node's substrings $s_P$ and $s_V$ in lines 5-6 (see
Definition~\ref{def:dynint}). In lines 7-10 we check for the base case
of the recursion, which occurs when all discriminative bytes are
exhausted and $K$ cannot be partitioned further. In this case, all
remaining pairs $(k,r) \in K$ have the same key $k$.  Leaf node $n$
contains the references to nodes in the database with this particular
key $k$.  In lines 11 we check if $K$ can be partitioned in dimension
$D$. If this is not the case, since all keys have the same value in
dimension $D$, we $\psi$-partition $K$ in the alternate dimension
$\overline{D}$.  Finally, in lines 14-16 we iterate over all non-empty
partitions in $M$ and recursively call the algorithm for each
partition $M[b]$, alternating dimension $D$ in round-robin fashion.

\begin{algorithm2e}[htbp] \scriptsize
  Let $n$ be a new RCAS node\;
  Let $(k_i, r_i)$ be the first key in $K$\;
  $g'_P \gets \textsf{dsc\_inc}(K, P, g_P)$\;
  $g'_V \gets \textsf{dsc\_inc}(K, V, g_V)$\;
  $n.s_P \gets k_i.P[g_P, g'_P - 1]$\;
  $n.s_V \gets k_i.V[g_V, g'_V - 1]$\;
  \If (\tcc*[f]{\textnormal{$n$ is a leaf}}) {
    $g'_P > \textsf{len}(k_i.P) \wedge g'_V > \textsf{len}(k_i.V)$}
  {
    $n.D \gets \bot$\;
    \lFor{$(k_j,r_j) \in K$} {
      append $r_j$ to $n.\text{refs}$
    }
    \Return{$n$}\;
  }
  \lIf{\textnormal{$g'_D > \textsf{len}(k_i.D)$}}{
    $D \gets \overline{D}$
  }
  $n.D \gets D$\;
  $M \gets \psi(K,D, g'_D)$\;
  \For{\textnormal{$b \gets \texttt{0x00}$ \KwTo $\texttt{0xFF}$}} {
    \If{\textnormal{partition $M[b]$ is not empty}} {
      $n.\text{children}[b] \gets \textsf{BulkLoadRCAS}(
        M[b], \overline{D}, g'_P, g'_V)$\;
    }
  }
  \Return{$n$}
  \caption{\textsf{BulkLoadRCAS}($K$, $D$, $g_P$, $g_V$)}
  \label{alg:bulk}
\end{algorithm2e}

\begin{lemma}\label{lemma:bulk}
  Let $K$ be a set of composite keys and let
  $l = \max_{k \in K} \{\textsf{len}(k.P) + \textsf{len}(k.V) \}$ be
  the length of the longest key. The time complexity of Algorithm
  \ref{alg:bulk} is $O\left(l \cdot |K|\right)$.
\end{lemma}

\begin{proof}
  We split the computations performed in function \textsf{BulkLoadRCAS}
  in Algorithm \ref{alg:bulk} into two groups. The first group includes the
  computations of the discriminative bytes across \emph{all} recursive
  invocations of \textsf{BulkLoadRCAS} (lines 1--12).  The second
  group consists of the $\psi$-partitioning  (line 13) across
  \emph{all} recursive invocations of \textsf{BulkLoadRCAS}.

  \emph{Group 1}: \textsf{BulkLoadRCAS} exploits the monotonicity
  of the discriminative bytes (Lemma \ref{lemma:monotonicity}) and
  passes the lower bound $g$ to function $\textsf{dsc\_inc}(K, D,
  g)$. As a result, we scan each byte of $k.P$ and $k.V$ only once for
  each $k$ in $K$ to determine the discriminative bytes.  This amounts
  to one full scan over all bytes of all keys in $K$ across all
  invocations of \textsf{BulkLoadRCAS}.  The complexity of this group
  is
  $O\left(\sum_{k \in K}(\textsf{len}(k.P) + \textsf{len}(k.V))
  \right) = O\left(l \cdot |K| \right)$.

  \emph{Group 2}: Given the position $g$ of the discriminative byte
  computed earlier, $\psi(K,D,g)$ must only look at the value of this
  byte in dimension $D$ of each key $(k,r) \in K$ and append $(k,r)$ to the
  proper partition $M[k.D[g]]$ in constant time.  Thus, a single
  invocation of $\psi(K,D)$ can be performed in $O(|K|)$ time.  The
  partitioning $\psi(K,D)$ is disjoint and complete (see Definition
  \ref{def:partitioning}), i.e., $|K| = \sum_{K_i \in \psi(K,D)}
  |K_i|$. Therefore, on each level of the RCAS index at most $|K|$
  keys need to be partitioned, with a cost of $O(|K|)$.  In the worst
  case, the height of the RCAS index is $l$, in which case every
  single path and value byte of the longest key is discriminative.
  Therefore, the cost of partitioning $K$ across all levels of the
  index is $O(l \cdot |K|)$.

  Although we partition $K$ recursively for every discriminative byte,
  the partitions become smaller and smaller and on each level add up
  to at most $|K|$ keys. Thus, the costs of the operations in group 1
  and group 2 add up to $O(2 \cdot l \cdot |K|) = O(l \cdot |K|)$.
\end{proof}

The factor $l$ in the complexity of Algorithm~\ref{alg:bulk} is
typically much smaller than $\textsf{len}(k.P) + \textsf{len}(k.V)$ of
the longest key $k$. For instance, assuming a combined length of just
six bytes would already give us around 280 trillion potentially
different keys.  So, we would need a huge number of keys for every
byte to become a discriminative byte on each recursion level.

\subsection{Querying RCAS}
\label{sec:querying}

Algorithm~\ref{alg:query} shows the pseudocode for evaluating a CAS
query on an RCAS index. The function \textsf{CasQuery} gets called
with the current node $n$ (initially the root node of the index), a
path predicate consisting of a query path $q$, and a range
$[v_l, v_h]$ for the value predicate. Furthermore, we need two buffers
$\texttt{buff}_{P}$ and $\texttt{buff}_{V}$ (initially empty) that hold, respectively,
all path and value bytes from the root to the current node $n$.
Finally, we require state information $s$ to
evaluate the path and value predicates (we
provide details as we go along) and an answer set $W$ to collect the
results.\footnote{The parameters $n$, $W$, $q$, and $[v_l,v_h]$ are
  call-by-reference, the parameters $\texttt{buff}_{V}$,
  $\texttt{buff}_{P}$, and $s$ are call-by-value.}

\begin{algorithm2e}[htb]
\scriptsize
\SetInd{0.75em}{0.75em} 
\DontPrintSemicolon
\SetKwProg{Fn}{Function}{}{end}
\SetKw{KwNot}{not}
\SetKw{KwAnd}{and}
\SetArgSty{textnormal}
$\texttt{UpdateBuffers}(n,\texttt{buff}_V,\texttt{buff}_P)$\;
$\texttt{match}_V \gets \texttt{MatchValue}(\texttt{buff}_V,v_l,v_h,s,n)$\;
$\texttt{match}_P \gets \texttt{MatchPath}(\texttt{buff}_P,q,s,n)$\;
\uIf{$\texttt{match}_V = \texttt{MATCH}$ \KwAnd
  $\texttt{match}_P = \texttt{MATCH}$} {
  $\textsf{Collect}(n,W)$\;
}
\ElseIf{$\texttt{match}_V \neq \texttt{MISMATCH}$ \KwAnd
  $\texttt{match}_P \neq \texttt{MISMATCH}$} {
  \For{each matching child $c$ in $n$} {
    $s' \gets \texttt{Update}(s)$\;
    $\textsf{CasQuery}(c,q,[v_l,v_h],\texttt{buff}_V,\texttt{buff}_P,s',W)$\;
  }
}
\caption{$\textsf{CasQuery}(n, q, [v_l, v_h], \texttt{buff}_V,
\texttt{buff}_P, s, W)$}
\label{alg:query}
\end{algorithm2e}

First, we update the buffers $\texttt{buff}_{V}$ and
$\texttt{buff}_{P}$, adding the information in the fields $s_V$ and
$s_P$ of the current node $n$ (line 1).  Next, we match the query
predicates to the current node. Matching values (line 2) works
differently to matching paths (line 3), so we look at the two cases
separately.

To match the current (partial) value \texttt{buff}$_V$ against the
value range $[v_l,v_h]$, their byte strings must be binary comparable
(for a detailed definition of binary-comparability see \cite{VL13}).
Function \texttt{MatchValue} proceeds as follows. We compute the
longest common prefix between $\texttt{buff}_{V}$ and $v_l$ and
between $\texttt{buff}_{V}$ and $v_h$. We denote the position of the
first byte for which $\texttt{buff}_{V}$ and $v_l$ differ by
\texttt{lo} and the position of the first byte for which
$\texttt{buff}_{V}$ and $v_h$ differ by \texttt{hi}.  If
$\texttt{buff}_{V}[\texttt{lo}] < v_l[\texttt{lo}]$, we know that the
node's value lies outside of the range and we return
\texttt{MISMATCH}. Similarly, if $\texttt{buff}_{V}[\texttt{hi}] >
v_h[\texttt{hi}]$, the node's value lies outside of the upper bound
and we return \texttt{MISMATCH} as well.  If $n$ is a leaf node and
$v_l \leq \texttt{buff}_{V} \leq v_h$, we return \texttt{MATCH}. If
$n$ is not a leaf node and $v_l[\texttt{lo}] <
\texttt{buff}_{V}[\texttt{lo}]$ and $\texttt{buff}_{V}[\texttt{hi}] <
v_h[\texttt{hi}]$, we know that all values in the
subtree rooted at $n$ match and we also return \texttt{MATCH}. In all
other cases we cannot make a decision yet and return
\texttt{INCOMPLETE}. The values of \texttt{lo} and \texttt{hi} are
kept in the state to avoid recomputing the longest common prefix from
scratch for each node. Instead we can resume the search from the
previous values of \texttt{lo} and \texttt{hi}.

Function \texttt{MatchPath} matches the query path $q$ against the
current path prefix $\texttt{buff}_P$. It supports the wildcard symbol
\texttt{*} and the descendant-or-self axis \texttt{//} that match any
child and descendant node, respectively. As long as we do not
encounter any wildcards in the query path $q$, we directly compare (a
prefix of) $q$ with the current content of $\texttt{buff}_{P}$ byte by
byte. As soon as a byte does not match, we return \texttt{MISMATCH}.
If we are able to successfully match the complete query path $q$
against a complete path in $\texttt{buff}_{P}$ (both terminated by
\texttt{\$}), we return \texttt{MATCH}. Otherwise, we need to continue
and return \texttt{INCOMPLETE}. When we encounter a wildcard
\texttt{*} in $q$, we match it successfully to the corresponding label
in $\texttt{buff}_{P}$ and continue with the next label. A wildcard
\texttt{*} itself will not cause a mismatch (unless we try to match it
against the terminator \texttt{\$}), so we either return
\texttt{MATCH} if it is the final label in $q$ and $\texttt{buff}_{P}$
or \texttt{INCOMPLETE}.  Matching the descendant-axis \texttt{//} is
more complicated. We note the current position where we are in
$\texttt{buff}_{P}$ and continue matching the label after \texttt{//}
in $q$. If at any point we find a mismatch, we backtrack to the next
path separator after the noted position, thus skipping a label in
$\texttt{buff}_{P}$ and restarting the search from there.  Once
$\texttt{buff}_{P}$ contains a complete path, we can make a decision
between \texttt{MATCH} or \texttt{MISMATCH}.

The algorithm continues by checking the outcomes of the value and path
matching (lines 4 and 6). If both predicates match, we descend the
subtree and collect all references (line 5 and function
\textsf{Collect} in Algorithm~\ref{alg:collect}). If at least one of
the outcomes is \texttt{MISMATCH}, we immediately stop the search in
the current node, otherwise we continue recursively with the matching
children of $n$ (lines 7--9). Finding the matching children depends on
the dimension $n.D$ of $n$ and follows the same logic as described
above for \texttt{MatchValue} and \texttt{MatchPath}. If node
$n.D = P$ and we have seen a descendant axis in the query path, all
children of the current node match.

\begin{algorithm2e}[htb]
\scriptsize
\SetInd{0.75em}{0.75em} 
\DontPrintSemicolon
\SetKwProg{Fn}{Function}{}{end}
\SetKw{KwNot}{not}
\SetKw{KwAnd}{and}
\SetArgSty{textnormal}
\uIf{$n$ is a leaf} {
  add references $r$ in $n.\text{refs}$ to $W$\;
}
\Else {
  \For{each child $c$ in $n$} {
    $\textsf{Collect}(c,W)$\;
  }
}
\caption{$\textsf{Collect}(n, W)$}
\label{alg:collect}
\end{algorithm2e}

\begin{example}
  We consider an example CAS query with path
  $q=\ptstr{/bom/item//battery\$}$ and a value range from
  $v_l = 10^5 =$ \vtstr{00\,01\,86\,A0} to $v_h = 5 \cdot 10^5 =$
  \vtstr{00\,07\,A1\,20}.  We execute the query on the index
  depicted in Figure~\ref{fig:solution}.

  \begin{itemize}[leftmargin=15pt]

  \item Starting at the root node $\mathsf{n}_1$, we load \vtstr{00}
    and \ptstr{/bom/item/ca} into $\texttt{buff}_V$ and
    $\texttt{buff}_P$, respectively. Function \texttt{MatchValue}
    matches \vtstr{00} and returns \texttt{INCOMPLETE}.
    \texttt{MatchPath} also returns \texttt{INCOMPLETE}: even though
    it matches \ptstr{/bom/item}, the partial label \ptstr{ca} does
    not match \texttt{battery}, so \ptstr{ca} is skipped by the
    descendant axis. Since both functions return \texttt{INCOMPLETE},
    we have to traverse all matching children. Since $\mathsf{n}_1$ is
    a value node ($\mathsf{n}_1.D = V$), we look for all matching
    children whose discriminative value byte is between \vtstr{01} and
    \vtstr{07}.  Nodes $\mathsf{n}_8$ and $\mathsf{n}_9$ satisfy this
    condition.

  \item Node $\mathsf{n}_8$ is a leaf. $\texttt{buff}_P$ and
    $\texttt{buff}_V$ are updated and contain complete paths and
    values. Byte \vtstr{01} matches, but byte
    $\texttt{buff}_V[3] = \vtstr{0E} < \vtstr{86} = v_l[3]$. Thus,
    \texttt{MatchValue} returns a \texttt{MISMATCH}. So does
    \texttt{MatchPath}. The search discards $\mathsf{n}_8$.

  \item Next we look at node $\mathsf{n}_9$. We find that
    $v_l[2] < \texttt{buff}_V[2] < v_h[2]$, thus all values of
    $\mathsf{n}_9$'s descendants are within the bounds $v_l$ and
    $v_h$, and \texttt{MatchValue} returns \texttt{MATCH}.
    \texttt{MatchPath} skips the next bytes \ptstr{r/} due to the
    descendant axis and resumes matching from there. After skipping
    \ptstr{r/}, it returns \texttt{MATCH}, as \ptstr{battery\$}
    matches the query path until its end. Both predicates match,
    invoking \textsf{Collect} on $\mathsf{n}_9$, which traverses
    $\mathsf{n}_9$'s descendants $\mathsf{n}_{10}$ and
    $\mathsf{n}_{11}$ and adds references $\mathsf{r}_3$,
    $\mathsf{r}'_3$, and $\mathsf{r}_4$ to $W$.

  \end{itemize}
\end{example}

\rev{Twig queries \cite{RB17} with predicates on multiple attributes
are broken into smaller CAS queries. Each root-to-leaf branch of the
twig query is evaluated independently on an appropriate RCAS index and
the resulting sets $W$ are joined to produce the final result.  The
join requires that the references $r \in W$ contain structural
information about a node's position in the tree (e.g., an OrdPath
\cite{PN04} node-labeling scheme). A query optimizer can use our cost
model to choose which RCAS indexes are used in a query plan.}

\section{Experimental Evaluation}
\label{sec:experiments}

\subsection{Setup and Datasets}
\label{sec:exsetup}

\begin{table}[H]\centering\setlength{\tabcolsep}{4pt}
\caption{\rev{Dataset Statistics}}
\label{tab:datasets}
\rev{
{\scriptsize\begin{tabular}{cr|crrc}
  Dataset
  & \multicolumn{1}{c|}{Size}
  & \multicolumn{1}{c}{Attribute}
  & \multicolumn{1}{c}{No. of Keys}
  & \multicolumn{1}{c}{Unique Keys}
  & \multicolumn{1}{c}{Size of Keys}
    \\
  \hline
  ServerFarm
    & 3.0GB
    & size
    & 21'291'019
    &  9'345'668
    & 1.7GB
    \\
  XMark
    & 58.9GB
    & category
    &  60'272'422
    &   1'506'408
    &  3.3GB
    \\
  Amazon
    & 10.5GB
    & price
    &  6'707'397
    &  6'461'587
    &  0.8GB
\end{tabular}}
}
\end{table}

\begin{table*}
\caption{CAS queries with their result size and the number of keys that
match the path, respectively value  predicate.}
\label{tab:queries}
\vspace{-10pt}
\begin{center}
  {\scriptsize \begin{tabular}{|ll|rrrrrr|}%
  \multicolumn{2}{c}{Query $Q$}
    & \multicolumn{2}{c}{Result size $(\sigma)$}
    & \multicolumn{2}{c}{Matching paths $(\sigma_P)$}
    & \multicolumn{2}{c}{Matching values $(\sigma_V)$}
    \\
  \hline
  \multicolumn{8}{|c|}{\rev{Dataset: \textbf{ServerFarm:size}}} \\
  \hline
  $Q_1$ & $Q(\texttt{/usr/include//}, \texttt{@size} \ge 5000)$
    & 142253 & $(0.6\%)$ & 434564 & (2.0\%) & 7015066 & (32.9\%)
    \\
  $Q_2$ & $Q(\texttt{/usr/include//}, 3000 \leq \texttt{@size} \leq 4000)$
    &  46471 & $(0.2\%)$ & 434564 & (2.0\%) & 1086141 & (5.0\%)
    \\
  $Q_3$ & $Q(\texttt{/usr/lib//}, 0 \leq \texttt{@size} \leq 1000)$
    & 512497 & $(2.4\%)$  & 2277518 & (10.6\%) & 8403809 & (39.4\%)
    \\
  $Q_4$ & $Q(\texttt{/usr/share//Makefile}, 1000 \leq \texttt{@size}
    \leq 2000)$
    & 1193 & $(<0.1\%)$  & 6408 & $(<0.1\%)$ & 2494804 & (11.7\%)
    \\
  $Q_5$ & $Q(\texttt{/usr/share/doc//README}, 4000 \leq \texttt{@size}
    \leq5000)$
    & 521 & $(<0.1\%)$  & 24698 & (0.1\%) & 761513 & (3.5\%)
    \\
  $Q_6$ & $Q(\texttt{/etc//}, \texttt{@size} \ge 5000)$
    & 7292 & $(<0.1\%)$ & 97758 & (0.4\%) & 7015066 & (32.9\%)
    \\
  \hline
  \hline
  \multicolumn{8}{|c|}{\rev{Dataset: \textbf{XMark:category}}} \\
  \hline
  $Q_7$ & $Q(\texttt{/site/people//interest}, 0 \leq \texttt{@category} \leq 50000)$
    & 1910524 & $(3.1\%)$ & 19009723 & (31.5\%) & 6066546 & (10.0\%)
    \\
  $Q_8$ & $Q(\texttt{/site/regions/africa//}, 0 \leq \texttt{@category} \leq 50000)$
    & 104500 & $(0.1\%)$ &  1043247 & (1.7\%) & 6066546 & (10.0\%)
    \\
  \hline
  \hline
  \multicolumn{8}{|c|}{\rev{Dataset: \textbf{Amazon:price}}} \\
  \hline
  $Q_{9}$ & $Q(\texttt{/CellPhones\&Accessories//}, 10000 \leq
    \texttt{@price} \leq 20000)$
    & 2758 & $(< 0.1\%)$ & 291625 & $(4.3\%)$ & 324272 & $(4.8\%)$
    \\
  $Q_{10}$ & $Q(\texttt{/Clothing/Women/*/Sweaters//}, 7000 \leq
    \texttt{@price} \leq 10000)$
    & 239 & $(< 0.1\%)$ & 4654 & $(< 0.1\%)$ & 269936 & $(4.0\%)$
    \\
  \hline
\end{tabular}}
\end{center}
\end{table*}

We use a virtual Ubuntu 18.04 server with eight cores and 64GB of main
memory.  All algorithms are implemented in C++ by the same author and
were compiled with clang 6.0.0 using -O3.  Each reported runtime
measurement is the average of 100 runs. \rev{All indexes are kept in
  main memory. The code\footnote{\url{https://github.com/k13n/rcas}}
  and the
  datasets\footnote{\url{https://doi.org/10.5281/zenodo.3739263}} used
  in the experimental evaluation can be found online.}

\textbf{Datasets.} \rev{We use three datasets, the ServerFarm dataset
  that we collected ourselves, a product catalog with products from
  Amazon \cite{RH16}, and the synthetic XMark \cite{AS02} dataset at a
  scale factor of 500.} The ServerFarm dataset mirrors the file system
of 100 Linux servers.  For each server we installed a default set of
packages and randomly picked a subset of optional packages.  In total
there are 21 million files.  For each file we record the file's full
path, size, various timestamps (e.g., a file's change time
\texttt{ctime}), file type, extension etc.  \rev{The Amazon dataset
  \cite{RH16} contains products that are hierarchically categorized.}
For each experiment we index the paths in a dataset along with one of
its attributes. We use the notation \texttt{\$dataset:\$attribute} to
indicate which attribute in a dataset is indexed.  E.g.,
ServerFarm:size contains the path of each file in the ServerFarm
dataset along with its size. \rev{The datasets do not have to fit into
  main memory, but we assume that the indexed keys fit into main
  memory.  Table \ref{tab:datasets} shows a summary of the datasets.}

\textbf{Compared Approaches.} We compare our RCAS index based on
dynamic interleaving with the following approaches that can deal with
variable-length keys.  The path-value (PV) and value-path (VP)
concatenations are the two possible $c$-order curves \cite{SN17}.  The
next approach, termed ZO for $z$-order \cite{GM66,JO84}, maps
variable-length keys to a fixed length as proposed by Markl
\cite{VM99}. Each path label is mapped to a fixed length using a
surrogate function.  Since paths can have a different number of
labels, shorter paths are padded with empty labels to match the number
of labels in the longest path in the dataset.  The resulting paths
have a fixed length $l_P$ and are interleaved with values of length
$l_V$ such that $\lceil \sfrac{l_V}{l_P} \rceil$ value bytes are
followed by $\lceil \sfrac{l_P}{l_V} \rceil$ path bytes.  The
label-wise interleaving (LW) interleaves one byte of the value with
one label of the path.  We utilize our RCAS index to identify the
dimension of every byte of the variable-length interleaved keys.  The
same underlying data-structure is also used to store the keys
generated by PV, VP, and ZO.  Finally, we compare RCAS against the CAS
index in \cite{CM15} that builds a structural summary (DataGuide
\cite{RG97}) and a value index (B+ tree\footnote{We use the tlx B+
tree (\url{https://github.com/tlx/tlx}), used also by \cite{RB18,HZ18}
for comparisons.}) and joins them to answer CAS queries.  We term this
approach XML as it was proposed for XML databases.

\subsection{Impact of Datasets on RCAS's Structure}
\label{sec:expstructure}

In Figure \ref{exp:shape} we show how the shape (depth and width) of
the RCAS index adapts to the datasets.  Figure \ref{exp:shape}a shows
the distribution of the node depths in the RCAS index for the
ServerFarm:size dataset. Because of the trie-based structure not every
root-to-leaf path in RCAS has the same length (see also Figure
\ref{fig:solution}).  The deepest nodes occur on level 33, but most
nodes occur on levels 10 to 15 \rev{with an average node depth of
  13.2}. This is due to the different lengths of the paths in a file
system. Figure \ref{exp:shape}b shows the number of nodes for each
node type. Recall from Section \ref{sec:nodelayout} that RCAS, like
ART \cite{VL13}, uses inner nodes with a physical fanout of 4, 16, 48,
and 256 pointers depending on the actual fanout of a node to reduce
the space consumption. The type of a node $n$ and its dimension $n.D$
are related.  Path nodes typically have a smaller fanout than value
nodes. This is to be expected, since paths only contain printable
ASCII characters (of which there are about 100), while values span the
entire available byte spectrum.  Therefore, the most frequent node
type for path nodes is type 4, while for value nodes it is type 256,
see Figure \ref{exp:shape}b. Leaf nodes do not partition the data and
thus their dimension is set to $\bot$ according to Definition
\ref{def:dynint}.  The RCAS index on the ServerFarm:size dataset
contains more path than value nodes.  This is because in this dataset
there are about 9M unique paths as opposed to about 230k unique
values.  Thus, the values contain fewer discriminative bytes than the
paths and can be better compressed by the trie structure.

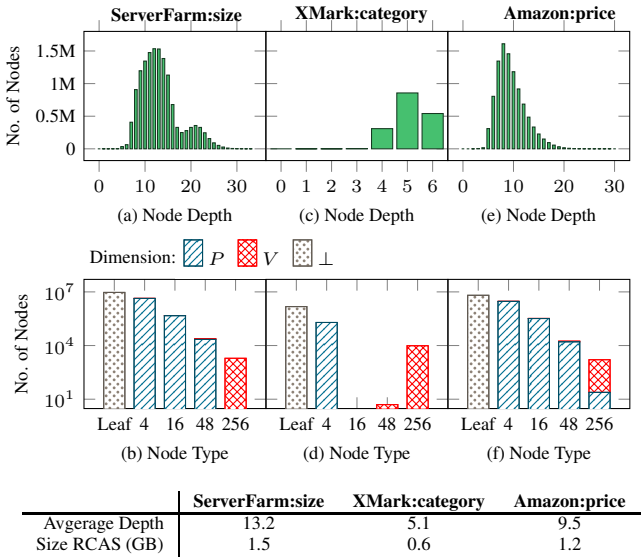
\begin{figure}[htb]
\begin{tikzpicture}
  \begin{groupplot}[
    col2,
    width=40mm,
    xticklabel style={
      align=center,
    },
    group style={
      columns=3,
      rows=2,
      horizontal sep=0pt,
      vertical sep=45pt,
      ylabels at=edge left,
      yticklabels at=edge left,
    },
    legend style={
      at={(0.5,1.00)},
      anchor=south west,
      legend columns=-1,
      draw=none,
      /tikz/every even column/.append style={column sep=3pt},
    },
  ]
  \node[anchor=center] at (rel axis cs:0.60,1.20)
    {\scriptsize \textbf{ServerFarm:size}};
  \node[anchor=center] at (rel axis cs:1.30,1.20)
    {\scriptsize \textbf{XMark:category}};
  \node[anchor=center] at (rel axis cs:2.05,1.20)
    {\scriptsize \textbf{Amazon:price}};
  \node[anchor=south west] at (rel axis cs:0.24,-0.91)
    {\scriptsize Dimension:};
  \nextgroupplot[
    ybar,
    bar width=1pt,
    ylabel={No. of Nodes},
    xlabel={(a) Node Depth},
    ytick scale label code/.code={},
    ytick={0,500000,1000000,1500000},
    yticklabels={0,0.5M,1M,1.5M},
    ymax=1800000,
    bar cycle list,
  ]
  \pgfplotstableread[col sep=comma]{experiments/bm_structure/depth_sf.txt}\tabA
  \addplot[barStructure] table[x=depth,y=nr_nodes,col sep=comma] {\tabA};
  %
  %
  %
  \nextgroupplot[
    ybar,
    bar width=8pt,
    xlabel={(c) Node Depth},
    xtick=data,
    ytick scale label code/.code={},
    ytick={0,500000,1000000,1500000},
    yticklabels={},
    ymax=1800000,
    bar cycle list,
  ]
  \pgfplotstableread[col sep=comma]{experiments/bm_structure/depth_xmark.txt}\tabC
  \addplot[barStructure] table[x=depth,y=nr_nodes,col sep=comma] {\tabC};
  %
  %
  \nextgroupplot[
    ybar,
    bar width=1pt,
    xlabel={\rev{(e) Node Depth}},
    ytick scale label code/.code={},
    ytick={0,500000,1000000,1500000},
    yticklabels={},
    ymax=1800000,
    bar cycle list,
  ]
  \pgfplotstableread[col sep=comma]{experiments/bm_structure/depth_amazon.txt}\tabC
  \addplot[barStructure] table[x=depth,y=nr_nodes,col sep=comma] {\tabC};
  \nextgroupplot[
    ybar stacked,
    bar width=8pt,
    ylabel={No. of Nodes},
    symbolic x coords={0,4,16,48,256},
    xticklabels={Leaf,4,16,48,256},
    xlabel={(b) Node Type},
    enlarge x limits=0.25,
    xtick=data,
    ymode=log,
    ymin=3,
    ymax=50000000,
    ybar legend,
  ]
  \pgfplotstableread[col sep=comma]{experiments/bm_structure/width_sf.txt}\tabB
  \addplot[barP] table[x=node_type,y=nr_p_nodes] {\tabB};
  \addplot[barV] table[x=node_type,y=nr_v_nodes] {\tabB};
  \addplot[barL] table[x=node_type,y=nr_l_nodes] {\tabB};
  \legend{$P$,$V$,$\bot$}
  \nextgroupplot[
    ybar stacked,
    bar width=8pt,
    symbolic x coords={0,4,16,48,256},
    xticklabels={Leaf,4,16,48,256},
    xlabel={(d) Node Type},
    enlarge x limits=0.25,
    xtick=data,
    ymode=log,
    ymin=3,
    ymax=50000000,
  ]
  \pgfplotstableread[col sep=comma]{experiments/bm_structure/width_xmark.txt}\tabD
  \addplot[barP] table[x=node_type,y=nr_p_nodes] {\tabD};
  \addplot[barV] table[x=node_type,y=nr_v_nodes] {\tabD};
  \addplot[barL] table[x=node_type,y=nr_l_nodes] {\tabD};
  \nextgroupplot[
    ybar stacked,
    bar width=8pt,
    symbolic x coords={0,4,16,48,256},
    xticklabels={Leaf,4,16,48,256},
    xlabel={\rev{(f) Node Type}},
    enlarge x limits=0.25,
    xtick=data,
    ymode=log,
    ymin=3,
    ymax=50000000,
  ]
  \pgfplotstableread[col sep=comma]{experiments/bm_structure/width_amazon.txt}\tabD
  \addplot[barP] table[x=node_type,y=nr_p_nodes] {\tabD};
  \addplot[barV] table[x=node_type,y=nr_v_nodes] {\tabD};
  \addplot[barL] table[x=node_type,y=nr_l_nodes] {\tabD};
  \end{groupplot}
\end{tikzpicture} \\[7pt]
\rev{{\hspace*{0.25cm} \scriptsize\begin{tabular}{r|ccc}
  & \textbf{ServerFarm:size}
  & \textbf{XMark:category}
  & \textbf{Amazon:price} \\
  \hline
  Avgerage~Depth & 13.2  &  5.1 & 9.5 \\
  Size RCAS (GB)  &  1.5  &  0.6 & 1.2 \\
\end{tabular}}}
\caption{Structure of the RCAS index}
\label{exp:shape}
\end{figure}

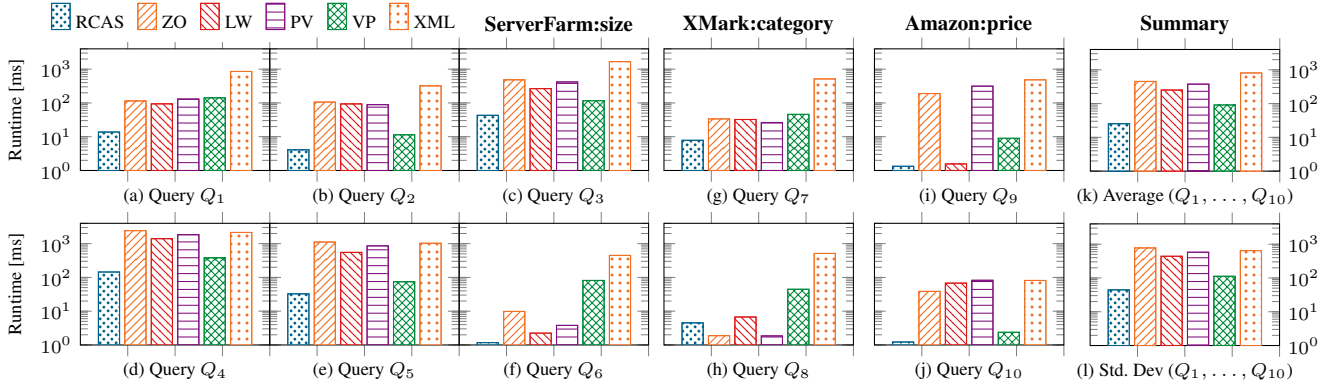
\begin{figure*}[htb]
\begin{tikzpicture}
  \begin{groupplot}[
    col2,
    width=41mm,
    height=32mm,
    ytick scale label code/.code={},
    xticklabel shift={-7pt},
    group style={
      columns=3,
      rows=2,
      horizontal sep=0pt,
      vertical sep=20pt,
      yticklabels at=edge left,
      ylabels at=edge left,
    },
    ybar,
    ymode=log,
    ymin=1,
    ymax=4000,
    ylabel={Runtime [ms]},
    legend style={
      at={(-0.2,1.05)},
      anchor=south west,
      legend columns=-1,
      draw=none,
      /tikz/every even column/.append style={column sep=3pt},
    },
  ]
  \node[anchor=south] at (rel axis cs:2.80,0.95)
    {\rev{\small \textbf{ServerFarm:size}}};
  \nextgroupplot[
    bar width=8pt,
    symbolic x coords={q},
    xticklabels={},
    xlabel={(a) Query $Q_1$},
    enlarge x limits=0.05,
  ]
  \pgfplotstableread[col sep=comma]{experiments/bm_query/q01.txt}\plot
  \addplot[barRCAS] table[x=query,y=dy]  {\plot};
  \addplot[barZO]   table[x=query,y=zo]  {\plot};
  \addplot[barLW]   table[x=query,y=lw]  {\plot};
  \addplot[barPV]   table[x=query,y=pv]  {\plot};
  \addplot[barVP]   table[x=query,y=vp]  {\plot};
  \addplot[barXML]  table[x=query,y=xml] {\plot};
  \legend{
    RCAS,
    ZO,
    LW,
    PV,
    VP,
    XML,
  }
  \nextgroupplot[
    bar width=8pt,
    symbolic x coords={q},
    xticklabels={},
    xlabel={(b) Query $Q_2$},
    enlarge x limits=0.05,
  ]
  \pgfplotstableread[col sep=comma]{experiments/bm_query/q02.txt}\plot
  \addplot[barRCAS] table[x=query,y=dy]  {\plot};
  \addplot[barZO]   table[x=query,y=zo]  {\plot};
  \addplot[barLW]   table[x=query,y=lw]  {\plot};
  \addplot[barPV]   table[x=query,y=pv]  {\plot};
  \addplot[barVP]   table[x=query,y=vp]  {\plot};
  \addplot[barXML]  table[x=query,y=xml] {\plot};
  \nextgroupplot[
    bar width=8pt,
    symbolic x coords={q},
    xticklabels={},
    xlabel={(c) Query $Q_3$},
    enlarge x limits=0.05,
  ]
  \pgfplotstableread[col sep=comma]{experiments/bm_query/q03.txt}\plot
  \addplot[barRCAS] table[x=query,y=dy]  {\plot};
  \addplot[barZO]   table[x=query,y=zo]  {\plot};
  \addplot[barLW]   table[x=query,y=lw]  {\plot};
  \addplot[barPV]   table[x=query,y=pv]  {\plot};
  \addplot[barVP]   table[x=query,y=vp]  {\plot};
  \addplot[barXML]  table[x=query,y=xml] {\plot};
  \nextgroupplot[
    bar width=8pt,
    symbolic x coords={q},
    xticklabels={},
    xlabel={(d) Query $Q_4$},
    enlarge x limits=0.05,
  ]
  \pgfplotstableread[col sep=comma]{experiments/bm_query/q04.txt}\plot
  \addplot[barRCAS] table[x=query,y=dy]  {\plot};
  \addplot[barZO]   table[x=query,y=zo]  {\plot};
  \addplot[barLW]   table[x=query,y=lw]  {\plot};
  \addplot[barPV]   table[x=query,y=pv]  {\plot};
  \addplot[barVP]   table[x=query,y=vp]  {\plot};
  \addplot[barXML]  table[x=query,y=xml] {\plot};
  \nextgroupplot[
    bar width=8pt,
    symbolic x coords={q},
    xticklabels={},
    xlabel={(e) Query $Q_5$},
    enlarge x limits=0.05,
  ]
  \pgfplotstableread[col sep=comma]{experiments/bm_query/q05.txt}\plot
  \addplot[barRCAS] table[x=query,y=dy]  {\plot};
  \addplot[barZO]   table[x=query,y=zo]  {\plot};
  \addplot[barLW]   table[x=query,y=lw]  {\plot};
  \addplot[barPV]   table[x=query,y=pv]  {\plot};
  \addplot[barVP]   table[x=query,y=vp]  {\plot};
  \addplot[barXML]  table[x=query,y=xml] {\plot};
  \nextgroupplot[
    bar width=8pt,
    symbolic x coords={q},
    xticklabels={},
    xlabel={(f) Query $Q_6$},
    enlarge x limits=0.05,
  ]
  \pgfplotstableread[col sep=comma]{experiments/bm_query/q06.txt}\plot
  \addplot[barRCAS] table[x=query,y=dy]  {\plot};
  \addplot[barZO]   table[x=query,y=zo]  {\plot};
  \addplot[barLW]   table[x=query,y=lw]  {\plot};
  \addplot[barPV]   table[x=query,y=pv]  {\plot};
  \addplot[barVP]   table[x=query,y=vp]  {\plot};
  \addplot[barXML]  table[x=query,y=xml] {\plot};
  \end{groupplot}
\end{tikzpicture}%
\hspace{-2mm}
\begin{tikzpicture}
  \begin{groupplot}[
    col2,
    width=41mm,
    height=32mm,
    ytick scale label code/.code={},
    xticklabel shift={-7pt},
    group style={
      columns=1,
      rows=2,
      vertical sep=20pt,
      horizontal sep=20pt,
      ylabels at=edge left,
    },
    ybar,
    ymode=log,
    ymin=1,
    ymax=4000,
    yticklabels={},
  ]
  \node[anchor=south] at (rel axis cs:1.00,1.05)
    {\rev{\small \textbf{XMark:category}}};
  \nextgroupplot[
    bar width=8pt,
    symbolic x coords={q},
    xticklabels={},
    xlabel={(g) Query $Q_7$},
    enlarge x limits=0.05,
  ]
  \pgfplotstableread[col sep=comma]{experiments/bm_query/q07.txt}\plot
  \addplot[barRCAS] table[x=query,y=dy]  {\plot};
  \addplot[barZO]   table[x=query,y=zo]  {\plot};
  \addplot[barLW]   table[x=query,y=lw]  {\plot};
  \addplot[barPV]   table[x=query,y=pv]  {\plot};
  \addplot[barVP]   table[x=query,y=vp]  {\plot};
  \addplot[barXML]  table[x=query,y=xml] {\plot};
  \nextgroupplot[
    bar width=8pt,
    symbolic x coords={q},
    xticklabels={},
    xlabel={(h) Query $Q_8$},
    enlarge x limits=0.05,
  ]
  \pgfplotstableread[col sep=comma]{experiments/bm_query/q08.txt}\plot
  \addplot[barRCAS] table[x=query,y=dy]  {\plot};
  \addplot[barZO]   table[x=query,y=zo]  {\plot};
  \addplot[barLW]   table[x=query,y=lw]  {\plot};
  \addplot[barPV]   table[x=query,y=pv]  {\plot};
  \addplot[barVP]   table[x=query,y=vp]  {\plot};
  \addplot[barXML]  table[x=query,y=xml] {\plot};
  \end{groupplot}
\end{tikzpicture}
\hspace{-2mm}
\begin{tikzpicture}
  \begin{groupplot}[
    col2,
    width=41mm,
    height=32mm,
    ytick scale label code/.code={},
    xticklabel shift={-7pt},
    group style={
      columns=1,
      rows=2,
      vertical sep=20pt,
      horizontal sep=20pt,
      ylabels at=edge left,
    },
    ybar,
    ymode=log,
    ymin=1,
    ymax=4000,
    yticklabels={},
    legend style={
      at={(0.0,1.05)},
      anchor=south west,
      legend columns=-1,
      draw=none,
      /tikz/every even column/.append style={column sep=3pt},
    },
  ]
  \node[anchor=south] at (rel axis cs:1.00,1.05)
    {\rev{\small \textbf{Amazon:price}}};
  \nextgroupplot[
    bar width=8pt,
    symbolic x coords={q},
    xticklabels={},
    xlabel={(i) Query $Q_9$},
    enlarge x limits=0.05,
  ]
  \pgfplotstableread[col sep=comma]{experiments/bm_query/q09.txt}\plot
  \addplot[barRCAS] table[x=query,y=dy]  {\plot};
  \addplot[barZO]   table[x=query,y=zo]  {\plot};
  \addplot[barLW]   table[x=query,y=lw]  {\plot};
  \addplot[barPV]   table[x=query,y=pv]  {\plot};
  \addplot[barVP]   table[x=query,y=vp]  {\plot};
  \addplot[barXML]  table[x=query,y=xml] {\plot};
  \nextgroupplot[
    bar width=8pt,
    symbolic x coords={q},
    xticklabels={},
    xlabel={(j) Query $Q_{10}$},
    enlarge x limits=0.05,
  ]
  \pgfplotstableread[col sep=comma]{experiments/bm_query/q10.txt}\plot
  \addplot[barRCAS] table[x=query,y=dy]  {\plot};
  \addplot[barZO]   table[x=query,y=zo]  {\plot};
  \addplot[barLW]   table[x=query,y=lw]  {\plot};
  \addplot[barPV]   table[x=query,y=pv]  {\plot};
  \addplot[barVP]   table[x=query,y=vp]  {\plot};
  \addplot[barXML]  table[x=query,y=xml] {\plot};
  \end{groupplot}
\end{tikzpicture}
\hspace{-2mm}
\begin{tikzpicture}
  \begin{groupplot}[
    col2,
    width=41mm,
    height=32mm,
    ytick scale label code/.code={},
    xticklabel shift={-7pt},
    group style={
      columns=1,
      rows=2,
      vertical sep=20pt,
      horizontal sep=20pt,
      ylabels at=edge left,
      yticklabels at=edge right,
    },
    ybar,
    ymode=log,
    ymin=1,
    ymax=4000,
    legend style={
      at={(0.0,1.05)},
      anchor=south west,
      legend columns=-1,
      draw=none,
      /tikz/every even column/.append style={column sep=3pt},
    },
  ]
  \node[anchor=south] at (rel axis cs:1.00,1.05)
    {\rev{\small \textbf{Summary}}};
  \nextgroupplot[
    bar width=8pt,
    symbolic x coords={avg},
    xticklabels={},
    xlabel={(k) Average ($Q_1,\ldots,Q_{10}$)},
    enlarge x limits=0.05,
  ]
  \pgfplotstableread[col sep=comma]{experiments/bm_query/rob_avg.txt}\plot
  \addplot[barRCAS] table[x=query,y=dy]  {\plot};
  \addplot[barZO]   table[x=query,y=zo]  {\plot};
  \addplot[barLW]   table[x=query,y=lw]  {\plot};
  \addplot[barPV]   table[x=query,y=pv]  {\plot};
  \addplot[barVP]   table[x=query,y=vp]  {\plot};
  \addplot[barXML]  table[x=query,y=xml] {\plot};
  \nextgroupplot[
    bar width=8pt,
    symbolic x coords={stdev},
    xticklabels={},
    xlabel={(l) Std.~Dev ($Q_1,\ldots,Q_{10}$)},
    enlarge x limits=0.05,
  ]
  \pgfplotstableread[col sep=comma]{experiments/bm_query/rob_stdev.txt}\plot
  \addplot[barRCAS] table[x=query,y=dy]  {\plot};
  \addplot[barZO]   table[x=query,y=zo]  {\plot};
  \addplot[barLW]   table[x=query,y=lw]  {\plot};
  \addplot[barPV]   table[x=query,y=pv]  {\plot};
  \addplot[barVP]   table[x=query,y=vp]  {\plot};
  \addplot[barXML]  table[x=query,y=xml] {\plot};
  \end{groupplot}
\end{tikzpicture}
%
\caption{(a)--(j): Runtime measurements for queries $Q_1$ to $Q_{10}$
in Table \ref{tab:queries}. (k)--(l): Average and standard deviation for queries $Q_1$ to
$Q_{10}$.}
\label{exp:rob_sfsize}
\end{figure*}

Figures \ref{exp:shape}c and \ref{exp:shape}d show the same
information for dataset XMark:category.  The RCAS index is more
shallow since there are only 7 unique paths and ca.~390k unique values
in a dataset of 60M keys. Thus the number of discriminative path and
value bytes is low and the index less deep.  Nodes of type 256 are
frequent (see Figure \ref{exp:shape}d) because of the larger number of
unique values.  While the XMark:category dataset contains 40M more
keys than the ServerFarm:size dataset, the RCAS index for the former
contains 1.5M nodes as compared to the 14M nodes for the latter.
\rev{The RCAS index for the Amazon:price dataset has similar
  characteristics as the ServerFarm:size dataset, see Figures
  \ref{exp:shape}e and \ref{exp:shape}f.}

\subsection{Robustness}
\label{sec:robustness}

Table \ref{tab:queries} shows a number of typical CAS queries with
their path and value predicates.  For example, query $Q_4$ looks for
all Makefiles underneath \texttt{/usr/share} that have a file size
between 1KB and 2KB.  In Table \ref{tab:queries} we report the
selectivity $\sigma$ of each query as well as path selectivity
$\sigma_P$ and value selectivity $\sigma_V$ of the queries' path and
value predicates, respectively.  The RCAS index avoids large
intermediate results that can be produced if an index prioritizes one
dimension over the other, or if it independently evaluates and joins
the results of the path and value predicates.  Query $Q_5$, e.g.,
returns merely 521 matches, but its path and value predicates return
intermediate results that are orders of magnitudes larger.

Figures \ref{exp:rob_sfsize}a to \ref{exp:rob_sfsize}f show that RCAS
consistently outperforms its competitors for queries $Q_1$ to $Q_6$
from Table \ref{tab:queries} on the ServerFarm:size dataset.
On these six queries
ZO and LW perform similarly as PV, which is
indicative for a high ``puff-pastry effect'' (see Section
\ref{sec:rw}) where one dimension is prioritized over another. In the
ServerFarm:size dataset ZO and LW prioritize the path dimension. The
reasons are twofold. The first reason is that the \texttt{size}
attribute is stored as a 64 bit integer since 32 bit integers cannot
cope with file sizes above $2^{32} \approx 4.3\text{GB}$.  The file
sizes in the dataset are heavily skewed towards small files (few bytes
or kilo-bytes) and thus have many leading zero-bytes. Many of these
most significant bytes do not partition the values. On the other hand,
the leading path bytes immediately partition the data: the second path
byte is discriminative since the top level of a file system contains
folders like \texttt{/usr}, \texttt{/etc}, \texttt{/var}. As a result,
ZO and LW fail to interleave at discriminative bytes and these
approaches degenerate to the level of PV. The second reason is
specific to ZO. To turn variable-length paths into fixed-length
strings, ZO maps path labels with a surrogate function and fills
shorter paths with trailing zero-bytes to match the length of the
longest path (see Section \ref{sec:exsetup}). We need 3 bytes per
label and the deepest path contains 21 labels, thus every path is
mapped to a length of 63 bytes and interleaved with the 8 bytes of the
values. Many paths in the dataset are shorter than the deepest path
and have many trailing zero-bytes. As explained earlier the values
(64-bit integers) have
many leading zero-bytes, thus the interleaved ZO string orders the
relevant path bytes before the relevant value bytes, further pushing
ZO towards PV.
Let us
look more closely at the results of query $Q_1$ in Figure
\ref{exp:rob_sfsize}a.
VP's runtime suffers because of the high value
selectivity $\sigma_V$.  XML performs badly because the intermediate
result sizes are one to two orders of magnitude larger than the final
result size.  RCAS with our dynamic interleaving is unaffected by the
puff-pastry effect in ZO, LW, and PV because it only interleaves at
discriminative bytes.  In RCAS the value selectivity (32\%) is
counter-balanced by the low path selectivity (2\%), thus avoiding VP's
pitfall. Lastly, RCAS does not materialize large intermediate results
as XML does.

Queries $Q_1$ and $Q_2$ have the same path predicate, but $Q_2$'s
value selectivity $\sigma_V$ is considerably lower. The query
performance of ZO, LW, and PV is unaffected by the lower $\sigma_V$
since $\sigma_P$ remains unchanged. This is further evidence that ZO
and LW prioritize the paths, just like PV. The runtime of VP and XML
benefit the most if the value selectivity $\sigma_V$ drops.  RCAS's
runtime still halves with respect to $Q_1$ and again shows the best
overall runtime.

Query $Q_3$ has the largest individual path and value selectivities,
and therefore produces the largest intermediate results. This is the
reason for its high query runtime.  RCAS has the best performance
since the final result size is an order of magnitude smaller than for
the individual path and value predicates.

Queries $Q_4$ and $Q_5$ look for particular files (\texttt{Makefile},
\texttt{README}) within large hierarchies.  Their low path selectivity
$\sigma_P$ should favor ZO, LW, and PV, but this is not the case.
Once Algorithm \ref{alg:query} encounters the descendant axis during
query processing, it needs to recursively traverse all children of
path nodes ($n.D = P$).  Fortunately, value nodes ($n.D = V)$ can
still be used to prune entire subtrees. Therefore, approaches that
alternate between discriminative path and values bytes, like RCAS, can
still effectively narrow down the search, even though the descendant
axis covers large parts of the index. Instead, approaches that
prioritize the paths (PV, ZO, LW) perform badly as they cannot prune
subtrees high up during query processing.

Query $Q_6$ has a very low path selectivity $\sigma_P$, but its value
selectivity $\sigma_V$ is high. This is the worst case for VP as it
can evaluate the path predicate only after traversing already a large
part of the index.  This query favors PV, LW, and ZO.  Nevertheless,
RCAS outperforms all other approaches.

The results for queries $Q_7$ and $Q_8$ on the XMark:category dataset
are shown in Figures \ref{exp:rob_sfsize}g and \ref{exp:rob_sfsize}h.
The gaps between the various approaches is smaller since the number of
unique paths is small (see Section \ref{sec:expstructure}). As a
result, the matching paths are quickly found by ZO, LW, and PV.
\rev{RCAS answers $Q_8$ in 4ms in comparison to 2ms for ZO and PV.}
VP performs worse on query $Q_8$ because of $Q_8$'s low $\sigma_P$ and
high $\sigma_V$.

\rev{Query $Q_9$ on dataset Amazon:price searches for all phones and
their accessories priced between \$100 and \$200. Selectivities
$\sigma_P$ and $\sigma_V$ are roughly 4.5\% whereas the total
selectivity is two orders of magnitude smaller. Figure
\ref{exp:rob_sfsize}i confirms that RCAS is the most efficient
approach.  Query $Q_{10}$ looks for all women's sweaters priced
between \$70 and \$100.  Sweaters exist for various child-categories
of category \texttt{Women}, e.g., \texttt{Clothing},
\texttt{Maternity}, etc. Query $Q_{10}$ uses the wildcard \texttt{*}
to match all of them.

Figures \ref{exp:rob_sfsize}k and \ref{exp:rob_sfsize}l show the
average runtime and the standard deviation for queries $Q_1$ to
$Q_{10}$.  RCAS is the most robust approach: it has the lowest average
runtime and standard deviation.

}

\subsection{\rev{Evaluation of Cost Model}}
\label{sec:expcostmodel}

\rev{This section uses the cost function
  $\widehat{C}(o, h, \phi, \varsigma_P, \varsigma_V)$ from
  Section~\ref{sec:costmodel} to estimate the cost of answering a
  query with RCAS.  We explain the required steps for query $Q_1$ on
  dataset ServerFarm:size (see Table \ref{tab:queries}).  First, we
  choose the alternating pattern of discriminative bytes in RCAS's
  dynamic interleaving by setting $\phi$ to $(V,P,V,P,\ldots)$.  For
  determining $h$ and $o$ we consider $|K|$ unique keys. Since each
  leaf represents a key, there are $o^h = |K|$ leaves. We set $h$ to
  the average depth of a node in the RCAS index (truncated to the next
  lower integer) and fanout $o$ to $\sqrt[h]{|K|}$. For dataset
  ServerFarm:size we have $|K| = 9.3\text{M}$ and $h=13$ (see Table
  \ref{tab:datasets} and Figure \ref{exp:shape}), thus
  $o = \sqrt[13]{9.3\text{M}} = 3.43$.  This is consistent with Figure
  \ref{exp:shape}a that shows that the most frequent node type in RCAS
  has a fanout of at most four.  Next we look at parameters
  $\varsigma_P$ and $\varsigma_V$. In our cost model, a query
  traverses a constant fraction $\varsigma_D$ of the children on each
  level of dimension $D$, corresponding to a selectivity of
  $\sigma_D = \varsigma_D \cdot \varsigma_D \cdot \ldots =
  \varsigma_D^N$ over all $N$ levels of dimension $D$.  $N$ is the
  number of discriminative bytes in dimension $D$. Thus, if a CAS
  query has a value selectivity of $\sigma_V$ we set
  $\varsigma_V = \sqrt[N]{\sigma_V}$. The value selectivity of query
  $Q_1$ is $\sigma_V = 32.9\%$ (see Table \ref{tab:queries}); the
  number of discriminative value bytes in $\phi$ is
  $N=\lceil \sfrac{h}{2} \rceil = \lceil \sfrac{13}{2} \rceil = 7$ (we
  use the ceiling since we start partitioning by $V$ in $\phi$), thus
  $\varsigma_V = \sqrt[7]{0.329}= 0.85$.  $\varsigma_P$ is determined
  likewise and yields
  $\varsigma_P = \sqrt[\lfloor \sfrac{13}{2} \rfloor]{0.02} = 0.52$
  for $Q_1$.

  We refine this cost model for path predicates containing the
  descendant axis \texttt{//} or the wildcard \texttt{*} followed by
  further path steps.  In such cases we use the path selectivity of
  the path predicate up to the first descendant axis or wildcard.  For
  example, for query $Q_4$ with path predicate
  \texttt{/usr/share//Makefile} and $\sigma_P = 0.03\%$, we use the
  path predicate \texttt{/usr/share//} with a selectivity of
  $\sigma_P = 44\%$.  This is so because the low path selectivity of
  the original predicate is not representative for the number of nodes
  that RCAS must traverse.  As soon as RCAS hits a descendant axis in
  the path predicate it can only use the value predicate to prune
  nodes during the traversal (see Section \ref{sec:querying}).

  \begin{table}[htb]
    \caption{\rev{Estimated cost $\widehat{C}$ and true cost ${C}$ for
        queries $Q_1$ to $Q_{10}$.}}
    \label{tab:costeval}
    \setlength\tabcolsep{4pt} \centering \hfill
    \begin{minipage}{0.49\linewidth}
      {\scriptsize%
        \begin{tabular}{c|S[table-format=6]S[table-format=6]S[table-format=1.2]}
          & {$\widehat{C}$}
          & {$C$}
          & {$E$}
          \\ \hline
          $Q_{1}$  & 105793 &  83190 & 1.27 \\
          $Q_{2}$  &  19157 &  28943 & 1.51 \\
          $Q_{3}$  & 542458 & 273824 & 1.98 \\
          $Q_{4}$  & 710128 & 784068 & 1.10 \\
          $Q_{5}$  & 111139 & 146124 & 1.31 \\
        \end{tabular}}
    \end{minipage}
    \hfill
    \begin{minipage}{0.49\linewidth}
      {\scriptsize%
        \begin{tabular}{c|S[table-format=6]S[table-format=6]S[table-format=1.2]}
          & {$\widehat{C}$}
          & {$C$}
          & {$E$}
          \\ \hline
          $Q_{6}$  &   9920 &   3062 & 3.24 \\
          $Q_{7}$  &  34513 &  30365 & 1.14 \\
          $Q_{8}$  &  18856 &  38247 & 2.03 \\
          $Q_{9}$  &  20421 &   4219 & 4.84 \\
          $Q_{10}$ &  17993 &  10698 & 1.68 \\
        \end{tabular}}
    \end{minipage}
    \hfill
  \end{table}

  In Table \ref{tab:costeval} we compare the estimated cost
  $\widehat{C}$ for the ten queries in Table \ref{tab:queries} with
  the true cost of these queries in RCAS. In addition to the estimated
  cost $\widehat{C}$ and the true cost $C$, we show the factor
  $E = \sfrac{\max(\widehat{C}, C)}{\min(\widehat{C}, C)}$ by which
  the estimate is off.  On average the cost model and RCAS are off by
  only a factor of two.

  Figure~\ref{fig:extuning} illustrates that the default values we
  have chosen for the parameters of $\widehat{C}$ yield near optimal
  results in terms of minimizing the average error $\overline{E}$ for queries
  $Q_1$ to $Q_{10}$. On the x-axis, we plot the deviation from the
  default value of a parameter. The values for $o$ and $h$ are spot
  on.  We overestimate the true path and value selectivities by a
  small margin; decreasing $\Delta\varsigma_P$ and $\Delta\varsigma_V$
  by 0.04 improves the error marginally.

\begin{figure}[htb]
\begin{tikzpicture}
  \begin{groupplot}[
    col2,
    height=28mm,
    width=34.5mm,
    group style={
      columns=4,
      rows=1,
      horizontal sep=0pt,
      ylabels at=edge left,
      yticklabels at=edge left,
    },
    legend style={
      at={(0.0,1.05)},
      anchor=south west,
      legend columns=6,
      draw=none,
    },
    ylabel={Avg.~Error $\overline{E}$},
    ymin=0,
    ymax=10,
  ]
  \nextgroupplot[
    xlabel={(a) $\Delta o$},
    xmin=-1.25,
    xmax=+1.25,
    xtick={-1,0,1},
  ]
  \pgfplotstableread[col sep=comma]{experiments/bm_calibration/o.txt}\plot
  \addplot[rcas]  table[x=delta_o,y=rel_error] {\plot};
  \nextgroupplot[
    xlabel={(b) $\Delta h$},
    xmin=-1.75,
    xmax=+1.75,
  ]
  \pgfplotstableread[col sep=comma]{experiments/bm_calibration/h.txt}\plot
  \addplot[rcas]  table[x=delta_h,y=rel_error] {\plot};
  \nextgroupplot[
    xlabel={(c) $\Delta \varsigma_P$},
    xmin=-0.2,
    xmax=+0.2,
    xtick={-0.1,0,0.1},
  ]
  \pgfplotstableread[col sep=comma]{experiments/bm_calibration/sp.txt}\plot
  \addplot[rcas]  table[x=delta_sp,y=rel_error] {\plot};
  \nextgroupplot[
    xlabel={(d) $\Delta \varsigma_V$},
    xmin=-0.2,
    xmax=+0.2,
    xtick={-0.1,0,0.1},
  ]
  \pgfplotstableread[col sep=comma]{experiments/bm_calibration/sv.txt}\plot
  \addplot[rcas]  table[x=delta_sv,y=rel_error] {\plot};
  \end{groupplot}
\end{tikzpicture}
\caption{Calibrating the cost model}
\label{fig:extuning}
\end{figure}
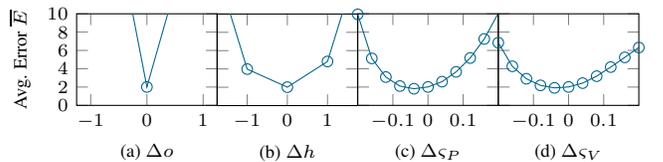
}

\subsection{Space Consumption and Scalability}
\label{sec:scalability}

\rev{Figure~\ref{exp:space} illustrates the space consumption of the
  indexes for our datasets.  RCAS, ZO, LW, PV, and VP all use the same
  underlying trie structure and have similar space requirements.  The
  XMark:category dataset can be compressed more effectively using
  tries because the number of unique paths and values is low (see
  Section \ref{sec:expstructure}) and common prefixes need to be
  stored only once. The trie indexes on Amazon:price do not compress
  the data as well as on the other two datasets since the lengthy
  titles of products do not share long common prefixes.  The XML index
  needs to maintain two structures, a DataGuide and a B+
  tree. Consequently, its space consumption is higher.  In addition,
  prefixes are not compressed as effectively in a B+ tree as they are
  in a trie. }

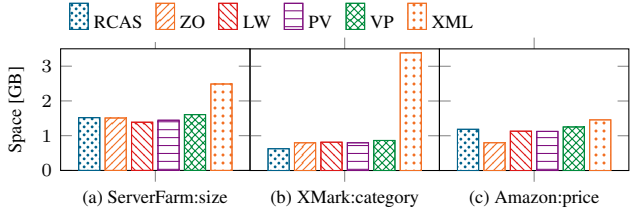
\begin{figure}[htb]
\begin{tikzpicture}
  \begin{groupplot}[
    col2,
    width=41mm,
    height=32mm,
    xtick=data,
    xticklabel style={
      align=center,
    },
    group style={
      columns=3,
      rows=1,
      horizontal sep=0pt,
      yticklabels at=edge left,
      ylabels at=edge left,
    },
    ybar,
    legend style={
      at={(0.0,1.05)},
      anchor=south west,
      legend columns=6,
      draw=none,
      /tikz/every even column/.append style={column sep=3pt},
    },
    ymin=1,
    ymax=3500000000,
    ylabel={Space [GB]},
    xticklabels={{{}}},
    ytick scale label code/.code={},
    xticklabel style = {yshift=+6pt},
    enlarge x limits=0.05,
  ]
  \nextgroupplot[
    bar width=8pt,
    xlabel={(a) ServerFarm:size},
    symbolic x coords={s,x,e},
    xticklabels={},
  ]
  \legend{
    RCAS,
    ZO,
    LW,
    PV,
    VP,
    XML,
  }
  \pgfplotstableread[col sep=comma]{experiments/bm_space/sf.txt}\tabA
  \addplot[barRCAS] table[x=dataset,y=dy]  {\tabA};
  \addplot[barZO]   table[x=dataset,y=zo]  {\tabA};
  \addplot[barLW]   table[x=dataset,y=lw]  {\tabA};
  \addplot[barPV]   table[x=dataset,y=pv]  {\tabA};
  \addplot[barVP]   table[x=dataset,y=vp]  {\tabA};
  \addplot[barXML]  table[x=dataset,y=xml] {\tabA};
  %
  %
  \nextgroupplot[
    bar width=8pt,
    xlabel={(b) XMark:category},
    xticklabels={},
    symbolic x coords={s,x,e},
    xticklabels={},
  ]
  \pgfplotstableread[col sep=comma]{experiments/bm_space/xmark.txt}\tabA
  \addplot[barRCAS] table[x=dataset,y=dy]  {\tabA};
  \addplot[barZO]   table[x=dataset,y=zo]  {\tabA};
  \addplot[barLW]   table[x=dataset,y=lw]  {\tabA};
  \addplot[barPV]   table[x=dataset,y=pv]  {\tabA};
  \addplot[barVP]   table[x=dataset,y=vp]  {\tabA};
  \addplot[barXML]  table[x=dataset,y=xml] {\tabA};
  %
  %
  \nextgroupplot[
    bar width=8pt,
    xlabel={\rev{(c) Amazon:price}},
    xticklabels={},
    symbolic x coords={s,x,e},
    xticklabels={},
  ]
  \pgfplotstableread[col sep=comma]{experiments/bm_space/amazon.txt}\tabA
  \addplot[barRCAS] table[x=dataset,y=dy]  {\tabA};
  \addplot[barZO]   table[x=dataset,y=zo]  {\tabA};
  \addplot[barLW]   table[x=dataset,y=lw]  {\tabA};
  \addplot[barPV]   table[x=dataset,y=pv]  {\tabA};
  \addplot[barVP]   table[x=dataset,y=vp]  {\tabA};
  \addplot[barXML]  table[x=dataset,y=xml] {\tabA};
  \end{groupplot}
\end{tikzpicture}
\caption{Space consumption}
\label{exp:space}
\end{figure}

In Figure~\ref{exp:scale} we analyze the scalability of the indexes in
terms of their space consumption and bulk-loading time as we increase the
number of
$(k,r)$ pairs to 100M.  We scale the ServerFarm:size dataset as needed
to reach a certain number of
$(k,r)$ pairs.  The space consumption (Figure~\ref{exp:scale}a) and
the index bulk-loading time (Figure~\ref{exp:scale}b) are linear in the
number of keys. The small drop for RCAS, ZO, LW, PV, and VP in
Figure~\ref{exp:scale}a is due to their trie structure.  These indexes
compress the keys more efficiently when we scale the dataset from
originally 21M keys to 100M keys (we do so by duplicating keys).  They
store the path $k.P$ and value $k.V$ only once for pairs
$(k,r_i)$ and $(k,r_j)$ with the same key
$k$. \rev{Figure \ref{exp:scale}b confirms that the time complexity of
  Algorithm~\ref{alg:bulk} to bulk-load the RCAS index is linear in
  the number of keys (see Lemma \ref{lemma:bulk}). Bulk-loading RCAS
  is a factor of two slower than bulk-loading indexes for static
  interleavings, but a factor of two faster than bulk-loading the XML
  index.  While RCAS takes longer to create the index it offers orders
  of magnitude better query performance as shown before.}
Figure~\ref{exp:scale} shows that the RCAS index for 100M keys requires 2GB
of memory and can be bulk-loaded in less than 5 minutes.

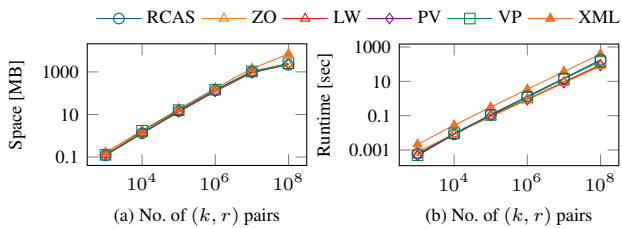
\begin{figure}[htb]
\begin{tikzpicture}
  \begin{groupplot}[
    col2,
    width=45mm,
    height=32mm,
    group style={
      columns=2,
      rows=2,
      horizontal sep=35pt,
    },
    legend style={
      at={(0.0,1.05)},
      anchor=south west,
      legend columns=6,
      draw=none,
    },
    xmode=log,
    ymode=log,
    xtick={10000, 1000000, 100000000},
  ]
  \nextgroupplot[
    xlabel={(a) No.~of $(k,r)$ pairs},
    ylabel={Space [MB]},
    ytick={0.001, 1, 1000},
    table/spacecol/.style={
      col sep=comma,
      x=size,
      y expr=\thisrow{#1}/1000000,
    },
    ytick={0.1,10,1000},
    yticklabels={0.1,10,1000},
  ]
  \addlegendimage{rcas}   \addlegendentry{RCAS}
  \addlegendimage{caszo}  \addlegendentry{ZO}
  \addlegendimage{caslw}  \addlegendentry{LW}
  \addlegendimage{caspv}  \addlegendentry{PV}
  \addlegendimage{casvp}  \addlegendentry{VP}
  \addlegendimage{casxml} \addlegendentry{XML}
  \pgfplotstableread[col sep=comma]{experiments/bm_scalability/space.txt}\plot
  \addplot[casxml] table[spacecol={xml}] {\plot};
  \addplot[casvp]  table[spacecol={vp}]  {\plot};
  \addplot[caspv]  table[spacecol={pv}]  {\plot};
  \addplot[caslw]  table[spacecol={lw}]  {\plot};
  \addplot[caszo]  table[spacecol={zo}]  {\plot};
  \addplot[rcas]   table[spacecol={dy}]  {\plot};
  %
  %
  \nextgroupplot[
    xlabel={(b) No.~of $(k,r)$ pairs},
    ylabel={Runtime [sec]},
    table/spacecol/.style={
      x=size,
      y expr=\thisrow{#1}/1000,
    },
    ytick={0.001,0.1,10,1000},
    yticklabels={0.001,0.1,10,1000},
  ]
  \pgfplotstableread[col sep=comma]{experiments/bm_scalability/time.txt}\plot
  \addplot[casxml] table[spacecol={xml}] {\plot};
  \addplot[casvp]  table[spacecol={vp}]  {\plot};
  \addplot[caspv]  table[spacecol={pv}]  {\plot};
  \addplot[caslw]  table[spacecol={lw}]  {\plot};
  \addplot[caszo]  table[spacecol={zo}]  {\plot};
  \addplot[rcas]   table[spacecol={dy}]  {\plot};
  \end{groupplot}
\end{tikzpicture}
\caption{Space consumption and bulk-loading time}
\label{exp:scale}
\end{figure}

\subsection{Summary}

We conclude with a summary of the key findings of our experimental
evaluation. First, RCAS shows the most robust query performance for a
wide set of CAS queries with varying selectivities: it exhibits the
most stable runtime in terms of average and standard deviation,
outperforming other state-of-the-art approaches by up to two orders of
magnitude. \rev{Second, our cost model yields a good estimate for the
true cost of a CAS query on the RCAS index.} Third, because of the
trie-based structure the RCAS index consumes less space than its
competitors, requiring only a third of the space used by a B+
tree-based approach.  The space consumption and bulk-loading time are
linear in the number of keys, allowing it to scale to a large number
of keys.


\section{Conclusion and Outlook}
\label{sec:conclusion}

We propose the Robust Content-and-Structure (RCAS) index for
semi-structured hierarchical data, offering a well-balanced
integration of paths and values in a single index. Our scheme avoids
prioritizing a particular dimension (paths or values), making the
index robust against queries with high individual selectivities
producing large intermediate results and a small final result. We
achieve this by developing a novel dynamic interleaving scheme that
exploits the properties of path and value attributes. Specifically, we
interleave paths and values at their \emph{discriminative bytes},
which allows our index to adapt to a given data distribution. In
addition to proving important theoretical properties, such as the
monotonicity of discriminative bytes and robustness, we show in our
experimental evaluation impressive gains: utilizing dynamic
interleaving our RCAS index outperforms state-of-the-art approaches by
up to two order of magnitudes on real-world and
synthetic datasets.

Future work points into several directions.  We plan to apply RCAS on
the largest archive of software source code, the Software Heritage
Dataset \cite{RC17,AP19}.  On the technical side we are working on
supporting incremental insertions and deletions in the RCAS index.
\rev{It would also be interesting to explore a disk-based RCAS index
based on a disk-based trie \cite{NA09}.} Further, we consider making
path predicates more expressive by, e.g., allowing arbitrary regular
expressions as studied by Baeza-Yates et al.~\cite{GG96}. \rev{On the
theoretical side it would be interesting to investigate the length of
the dynamic interleavings for different data distributions.}



\balance
\bibliographystyle{abbrv}
\bibliography{bibliography}

\end{document}